\documentclass[journal,twoside,web]{ieeecolor}
\usepackage{generic}

\usepackage{amssymb,amsmath,latexsym,amsfonts,amsthm,mathtools}

\usepackage{cite,graphicx}
\usepackage{float,subcaption}
\usepackage{hyperref}
\hypersetup{hidelinks}
\usepackage{textcomp}
\newtheorem{theorem}{Theorem}
\newtheorem{lemma}{Lemma}
\newtheorem{corollary}{Corollary}
\newtheorem{proposition}[theorem]{Proposition}
\newtheorem{problem}{Problem}

\theoremstyle{remark}
\newtheorem{remark}{Remark}
\newtheorem{assumption}{Assumption}

\theoremstyle{definition}

\usepackage{anyfontsize}

\usepackage{newtxmath}
\usepackage{booktabs}
\allowdisplaybreaks
\definecolor{mygreen}{rgb}{0, 0.5, 0}
\usepackage{hyperref}
\hypersetup{colorlinks, breaklinks, citecolor=blue, linkcolor=mygreen, urlcolor=blue}
\usepackage[nameinlink,noabbrev,sort&compress]{cleveref}
\Crefname{problem}{Problem}{Problems}
\Crefname{assumption}{Assumption}{Assumptions}
\newcommand{\abs}[1]{\lvert #1 \rvert}

\usepackage{tikz}
\usetikzlibrary{arrows.meta,positioning,fit,calc,backgrounds}

\def\BibTeX{{\rm B\kern-.05em{\sc i\kern-.025em b}\kern-.08em
	T\kern-.1667em\lower.7ex\hbox{E}\kern-.125emX}}
\begin{document}
\title{Admissibility-Preserving Control for Strict-Feedback Nonlinear Systems with Asymmetric Actuator Constraints}
\author{Saurabh Kumar, Shashi Ranjan Kumar, \IEEEmembership{Senior Member, IEEE}, and Abhinav Sinha, \IEEEmembership{Senior Member, IEEE}
	\thanks{S. Kumar and S.R. Kumar are with the Intelligent Systems and Control (ISaC) Lab, Department of Aerospace Engineering, Indian Institute of Technology Bombay, Maharashtra 400076 India (e-mails: \{saurabh.k,srk\}@aero.iitb.ac.in). }
	\thanks{A. Sinha is with the Guidance, Autonomy, Learning, and Control for Intelligent Systems (GALACxIS) Lab, Department of Aerospace Engineering and Engineering Mechanics, University of Cincinnati, OH 45221, USA. (e-mail: abhinav.sinha@uc.edu).}}

\maketitle
\begin{abstract}
This paper develops \emph{Admissibility-Preserving Control} (APC), a realization-centered safety-critical control framework for strict-feedback systems subject to asymmetric actuator limits, time-varying output constraints, and actuator-rate limitations. APC denotes the overall control architecture, whereas an \emph{Admissibility-Preserving Input Realization} (APIR) denotes its constraint-realization module. Therein, the APIR dynamically generates the physical plant input while rendering its prescribed asymmetric actuator set forward invariant. In contrast to algebraic clipping and post-design saturation compensation, the actuator limits are embedded directly in a continuously differentiable dynamic realization with user-selectable regularity and interpretable tuning parameters. The APIR is integrated with recursive backstepping by treating the realized plant input as an additional state. The resulting design does not require an input-to-state stability assumption on the uncontrolled plant. Instead, the nonlinear drift terms are compensated recursively, subject to an explicit compatibility condition between the desired motion, the available control authority, and the APIR interior gain. The framework is further extended to time-varying output-safe tracking through a smooth asymmetric logarithmic barrier coordinate and its associated Lyapunov function and to simultaneous actuator-magnitude and rate constraints through a cascaded APIR. Rigorous Lyapunov and invariance analyses establish regional asymptotic tracking, forward invariance of the compatible admissible sets, and boundedness of all closed-loop signals. Numerical studies illustrate asymmetric actuator utilization, output-safety preservation, and magnitude-rate constraint enforcement.
\end{abstract}

\section{Introduction}
\label{sec:introduction}
\IEEEPARstart{S}{afety}-critical control systems must achieve their regulation or tracking objectives without driving the physical system outside its admissible operating region. In such systems, actuator magnitude and rate limits, state envelopes, and output restrictions are integral to safe operation rather than secondary implementation details. Ignoring these limits can degrade closed-loop performance, invalidate nominal stability guarantees, or expose the plant and actuator hardware to unsafe operating conditions. This motivates control architectures in which constraint satisfaction is incorporated into the closed-loop design from the outset.

Bounded-input stabilization has a long history. For linear systems, global asymptotic stabilization under bounded control is characterized by asymptotic null controllability\footnotemark \footnotetext{A linear system ($\dot{\mathbf{x}} =  \mathbf{Ax} + \mathbf{Bu}$) is said to be asymptotically null controllable if and only if (i) no eigenvalues of $\mathbf{A}$ are in the right half of $s$-plane (ii) The pair ($\mathbf{A}, \mathbf{B}$) is stabilizable in the ordinary sense, that is, all the uncontrollable modes of the system have negative real parts.} \cite{203432}, while linear feedback cannot globally stabilize chains of three or more integrators under bounded inputs \cite{doi:10.1080/00207176908905846,261255}. Low-gain, high--low-gain, parametric Lyapunov/Riccati, and nested-saturation designs provide important global or semi-global solutions for several linear classes \cite{doi:10.1016/0167-6911(93)90033-3,doi:10.1002/rnc.4590050503,486638,4610041,doi:10.1016/j.jfranklin.2020.08.025,doi:10.1016/0167-6911(92)90001-9,362853}. Adaptive and auxiliary-system approaches have also been used to compensate for the mismatch between nominal and saturated inputs \cite{doi:10.1016/j.automatica.2005.02.009,728887,doi:10.1016/S0005-1098(00)00154-0,doi:10.1016/0005-1098(94)90202-X,333787,doi:10.1016/0005-1098(95)00059-6}. These developments establish a substantial foundation, but many are centered on linear dynamics, symmetric bounds, or compensation added after the nominal feedback law has been constructed.

For nonlinear systems, saturation-aware backstepping and adaptive designs have employed smooth approximations, Nussbaum-type mechanisms, neural approximators, disturbance observers, and auxiliary dynamics \cite{5723705,6975243}. The majority of these designs nevertheless adopt symmetric actuator limits. In practice, positive and negative control authority can differ markedly because of thrust/braking asymmetry, aerodynamic effectiveness, heating and cooling limits, or electromechanical actuation characteristics. Replacing asymmetric limits with a symmetric interval based on the smaller magnitude discards usable authority in the less restrictive direction. Although the practical significance of asymmetric saturation has been emphasized \cite{doi:10.1080/00207721.2021.1989726}, explicit nonlinear designs that treat unequal bounds directly remain comparatively limited \cite{6875955,7428909,6994268}.

Output and state constraints introduce an additional safety layer. Barrier Lyapunov functions provide a systematic mechanism for preserving constrained regions in strict-feedback systems \cite{doi:10.1016/j.automatica.2008.11.017}. Prescribed-performance and funnel-control methods instead organize the design around a tracking-error envelope and have also been extended to input-constrained settings \cite{4639441,doi:10.1051/cocv:2002064,berger2021funnel,10273607,10004950,https://doi.org/10.1002/rnc.4466,10.1016/j.sysconle.2004.05.014,10387582,10.1016/j.arcontrol.2025.101024}. These approaches are complementary to the present work: funnel and prescribed-performance designs primarily shape the tracking-error evolution, whereas the objective here is to make admissibility of the \emph{realized physical input} an intrinsic property of a dynamic feedback realization and then compose that realization with output-set enforcement.

Motivated by this distinction, we use the term \emph{Admissibility-Preserving Control} (APC) for a realization-centered control principle in which constrained physical variables are generated through dynamics whose admissible sets are invariant by construction, and the feedback law is synthesized around those realized variables. The corresponding dynamic module is termed an \emph{Admissibility-Preserving Input Realization} (APIR). Thus, APC denotes the control architecture, whereas APIR denotes the realization mechanism through which actuator admissibility is embedded into the closed-loop dynamics. In the present strict-feedback construction, the APIR generates the physical plant input through a nonlinear differential equation whose effective gain decreases near the prescribed actuator limits and whose vector field points inward at the boundaries. Static saturation already guarantees algebraic boundedness of the clipped signal. The proposed APIR addresses a different need by providing a continuously differentiable dynamic state, strict interior admissibility, a tunable boundary margin, and direct compatibility with recursive Lyapunov design.

The APIR is incorporated into backstepping by augmenting the realized plant input as an additional state. The nonlinear drift terms are compensated recursively, so the design does not require an input-to-state stability assumption on the uncontrolled plant. This removal of an open-loop ISS hypothesis does not eliminate the finite-authority limitation inherent to bounded control. Accordingly, the closed-loop results are stated over compatible Lyapunov sublevel sets for which the desired terminal virtual input is feasible and the APIR gain remains uniformly separated from zero. This formulation distinguishes the absence of an ISS assumption from the control-authority and reference-feasibility conditions that any hard-bounded controller must satisfy.

The framework is extended in two directions. First, a time-varying barrier Lyapunov function is combined with APC so that the output remains in a prescribed safe set while the realized input remains inside its asymmetric actuator set. Second, a two-layer APIR enforces simultaneous actuator-magnitude and actuator-rate constraints. A general cascade is subsequently analyzed at the realization level, where each internal layer and each layer-output rate are kept admissible. This provides a foundation for higher-order physical actuator constraints once the corresponding derivative mappings are verified. The main contributions are summarized as follows.
\begin{itemize}
	\item A continuously differentiable asymmetric APIR is developed whose admissible actuator interval is forward invariant under bounded commands. The analysis also establishes a compact invariant interior set, a positive lower bound on the APIR gain, the symmetric case as a special case, and the exact finite regularity of the asymmetric realization.
	\item The APIR retains the distinct lower and upper actuator limits rather than conservatively symmetrizing them. Its parameters admit an interpretable design: $p_{1}$ regulates the response time for constant-command operation, $p_{2}$ determines the equilibrium boundary margin through an exact algebraic relation, and $\gamma$ governs the gain-rolloff profile, differentiability order, and equilibrium utilization.
	\item An APC backstepping law is constructed by treating the realized plant input as an additional dynamic state. Unlike post-design compensation architectures, admissibility is embedded in the augmented feedback dynamics. The stability result avoids an open-loop ISS assumption while explicitly accounting for model knowledge, reference feasibility, and finite actuator authority.
	\item APC is combined with a time-varying barrier Lyapunov function to enforce simultaneous actuator and output admissibility. The result is formulated over a compatible full-state sublevel set, thereby accounting for the higher-relative-degree viability requirements that cannot be captured by output membership alone.
	\item A cascaded APIR is developed for simultaneous asymmetric magnitude and rate constraints. The general $r$-layer analysis proves admissibility of every realization state and its layer-output rate, while clarifying the additional verification needed before interpreting deeper layers as prescribed higher derivatives of the physical plant input.
	\item \emph{APC compatibility} is introduced as the feasibility principle that reconciles the selected closed-loop region and desired motion with finite actuator authority and a uniformly positive APIR interior gain. Together with Lyapunov and continuation arguments, this condition yields regional asymptotic tracking, forward invariance of the compatible input and output sets, nonsingularity of the commanded-input maps, and boundedness of all closed-loop signals.
\end{itemize}

\section{Problem Formulation}\label{sec:problem_formulation}
Consider the following class of nonlinear strict-feedback systems
\begin{subequations}\label{eq:system_dynamics}
	\begin{align}
		\dot{x}_{i} =&~f_{i}\left(\bar{x}_{i}\right) + g_{i}\left(\bar{x}_{i}\right) x_{i+1}, \quad i = 1, 2, \ldots, n-1, \label{eq:xi_dot}\\
		\dot{x}_{n} =&~f_{n}\left(\bar{x}_{n}\right) + g_{n}\left(\bar{x}_{n}\right) u , \label{eq:xn_dot}\\
		y =&~ x_{1}, \label{eq:y}
	\end{align}
\end{subequations}
where $\bar{x}_i=[x_1,\ldots,x_i]^\top\in\mathbb{R}^i$ denotes the state vector, $x=[x_1,\ldots,x_n]^\top\in\mathbb{R}^n$ is the system state, $y\in\mathbb{R}$ is the measured output, $u\in\mathbb{R}$ is the actuator output applied to the plant, and $f_i(\cdot),g_i(\cdot):\mathbb{R}^i\rightarrow\mathbb{R}$ are smooth nonlinear functions.

The actuator is driven by a controller-generated command $v\in\mathbb{R}$. Owing to physical limitations, the actuator cannot reproduce arbitrary commanded inputs and is required to operate within prescribed asymmetric magnitude limits. Consequently, the actual plant input $u$ is constrained to remain inside the admissible set $\mathbb{U} \coloneqq \{u\in\mathbb{R}\mid u_{\min}<u<u_{\max}\}$, where $u_{\min}<0$ and $u_{\max}>0$ denote the lower and upper actuator limits, respectively.

A widely adopted representation of asymmetric actuator saturation is the static nonlinear mapping
\begin{equation}\label{eq:sat_def}
	u=\operatorname{sat}(v)=
	\begin{cases}
		u_{\max}, & v>u_{\max},\\
		v, & u_{\min}\le v\le u_{\max},\\
		u_{\min}, & v<u_{\min},
	\end{cases}
\end{equation}
which clips the controller output whenever it exceeds the admissible range. The map in \eqref{eq:sat_def} guarantees algebraic boundedness of the realized input, but it is nondifferentiable at $v=u_{\min}$ and $v=u_{\max}$, may place the input directly on the constraint boundary, and does not provide a tunable dynamic state for recursive design. Consequently, standard smooth backstepping cannot be applied to the clipped map without introducing an additional constraint-handling mechanism. The objective of the proposed APIR is therefore not to replace the elementary boundedness property of static saturation, but to realize the physical input through differentiable dynamics with strict interiority and a quantifiable gain margin.

These observations motivate a dynamic actuator realization that is compatible with recursive nonlinear synthesis and can be analyzed independently of the plant. The realization is used as a modular safety component: the single layer enforces actuator magnitude admissibility, a second layer enforces a compatible rate constraint, and a barrier Lyapunov construction enforces output admissibility. The general cascade developed later guarantees the admissibility of its internal realization states and layer-output rates. Interpreting deeper layers as higher physical derivatives requires an additional derivative-mapping verification.

The following assumptions are standard in recursive backstepping design and are introduced to facilitate controller synthesis and stability analysis \cite{5723705,6975243,6875955,7428909,6994268}.
\begin{assumption} \label{assum:desired_traj}
	The desired trajectory $y_d(t)$ is sufficiently smooth such that $y_d(t)$ and its derivatives up to order $(n+1)$ are bounded for all $t\ge0$.
\end{assumption}
\begin{assumption}[Model regularity and control effectiveness]\label{assum:control_effectiveness}
	For each $i=1,\ldots,n$, the functions $f_i(\bar{x}_i)$ and $g_i(\bar{x}_i)$ are known, available for controller synthesis, and sufficiently differentiable for the recursive construction. The input gain $g_i(\bar{x}_i)$ has a known constant sign and satisfies $0<\underline{g}_i\leq |g_i(\bar{x}_i)|\leq \overline{g}_i<\infty$ for all $\bar{x}_i\in\mathbb{R}^i$, where $\underline{g}_i$ and $\overline{g}_i$ are known positive constants.
\end{assumption}
\begin{remark}
	The controller explicitly cancels $f_i$ and divides by $g_i$; therefore, the present model-based result requires knowledge of these functions, not only the sign of $g_i$. The uniform lower bound on $\lvert g_{i} \rvert$ prevents loss of control effectiveness in the recursive virtual-control channels. Uncertain and learned models are outside the scope of the present analysis and constitute natural extensions of APC.
\end{remark}
\begin{remark}
	No input-to-state stability property is imposed on the uncontrolled plant. The drift nonlinearities are handled directly by the virtual stabilizing functions. This should not be interpreted as unrestricted global stabilizability under finite actuation. The desired terminal virtual input and the commanded-input realization must remain compatible with the available actuator authority. The corresponding compatibility condition is stated explicitly after the commanded control law is introduced.
\end{remark}
\begin{problem}[Admissibility-preserving output tracking]\label{problem1}
	For the nonlinear system \eqref{eq:system_dynamics}, construct an APC law such that the output asymptotically tracks $y_d(t)$ and the realized plant input satisfies $u(t)\in\mathbb{U}$ for all $t\geq0$, for every initial condition in a compatible closed-loop sublevel set.
\end{problem}
In safety-critical applications, input admissibility alone may be insufficient because the output must also remain inside an operational envelope throughout the time. This motivates the following extension.
\begin{problem}[Admissibility-preserving constrained output tracking]\label{problem2}
	Extend the controller of \Cref{problem1} so that, in addition to asymptotic tracking and $u(t)\in\mathbb{U}$, the output satisfies
	\begin{equation*}
		y(t)\in\mathbb{Y}(t)\coloneqq\left\{y\in\mathbb{R}\mid \underline{y}(t)<y<\overline{y}(t)\right\},\quad t\geq0,
	\end{equation*}
	where $\underline{y}(t)$ and $\overline{y}(t)$ are continuous and satisfy $\underline{y}(t)<\overline{y}(t)$. The guarantee is required for compatible initial full states, not solely for an initially admissible output.
\end{problem}
The two-layer extension developed later enforces actuator magnitude and rate constraints. The general $r$-layer construction establishes admissibility of the realization states and the output rate of each layer. It is presented as a modular foundation for higher-order actuator models rather than as an automatic identification of every auxiliary layer with a prescribed derivative of $u$.

The subsequent development proceeds in a modular manner. We first construct an actuator realization that guarantees forward invariance of the admissible input set independently of the plant dynamics. This realization is then integrated with a recursive backstepping controller to solve \Cref{problem1}. Finally, the proposed framework is systematically extended to accommodate time-varying output constraints and simultaneous actuator magnitude-rate constraints while preserving the stability properties of the overall closed-loop system.
\section{The APC Framework}\label{sec:main_results}
This section develops APC through three reusable elements: \emph{APIR realizes admissibility}, \emph{APC compatibility certifies feasible use of finite authority}, and \emph{APC synthesis achieves the control objective around the admissible realization}. We first establish the intrinsic admissibility and regularity properties of the APIR. The realization is then incorporated into recursive backstepping under the APC compatibility condition, composed with a barrier Lyapunov construction for output safety, and extended through a two-layer APIR for simultaneous magnitude and rate constraints.
\subsection{Admissibility-Preserving Input Realization}
The APIR is analyzed first as an independent dynamic module. Its admissibility guarantee depends on the boundedness of its command but not on the nonlinear plant. This separation permits the controller analysis to focus on a self-consistent command bound and a positive interior APIR gain.

In this paper, the APIR principle is instantiated by the following smooth asymmetric APIR family:
\begin{align} \label{eq:single_APIR}
	\dot{u} =&~ p_1 \left[\rho  \left\{ 1 - \left(\dfrac{u}{u_{\max}}\right)^{\gamma} \right\}  + \left(1 - \rho \right)  \left\{ 1 - \left(\dfrac{u}{u_{\min}}\right)^{\gamma} \right\}  \right] u_{c} \nonumber\\
	&~-p_1p_2 u,
\end{align}
where $u_c$ denotes the commanded control input generated by the controller, while $u$ denotes the realized actuator output supplied to the plant. The positive constants $p_1,p_2\in\mathbb{R}_+$ are design parameters, $\gamma=2n$ with $n\in\mathbb{N}$ is an even integer determining the smoothness order and transition profile of the realization, and $u_{\min}<0<u_{\max}$ denote the prescribed lower and upper actuator limits, respectively. The binary variable $\rho$ is defined as $\rho =1$ for $u>0$ and $\rho=0$ otherwise. The asymmetric APIR family is well defined for $u_{\min}<0<u_{\max}$. More broadly, APIR refers to the admissibility-preserving realization principle, while \eqref{eq:single_APIR} denotes the particular APIR family used throughout this paper.

\Cref{fig:saturation_model_block} illustrates the proposed APIR, where the controller generates the commanded input $u_c$, which is subsequently processed by the realization dynamics to produce the admissible actuator input $u$ supplied to the nonlinear plant.
\begin{figure}[!ht]
	\centering
	\resizebox{\linewidth}{!}{%
		\definecolor{APCNavy}{RGB}{30,55,82}
		\definecolor{APCBlue}{RGB}{64,105,145}
		\definecolor{APCBlueFill}{RGB}{229,238,247}
		\definecolor{APCTeal}{RGB}{47,120,121}
		\definecolor{APCTealFill}{RGB}{226,242,240}
		\definecolor{APCSlate}{RGB}{92,105,117}
		\definecolor{APCGrayFill}{RGB}{239,242,244}
		\definecolor{APCGroupFill}{RGB}{247,249,251}
		\definecolor{APCFeedback}{RGB}{146,78,66}
		\begin{tikzpicture}[
			>=Stealth,
			signal/.style={
				-{Stealth[length=2.5mm,width=1.7mm]},
				line width=0.85pt,
				draw=APCNavy
			},
			feedback/.style={
				-{Stealth[length=2.5mm,width=1.7mm]},
				line width=0.85pt,
				draw=APCFeedback
			},
			synthblock/.style={
				draw=APCBlue,
				fill=APCBlueFill,
				line width=0.9pt,
				rounded corners=2.2pt,
				minimum width=34mm,
				minimum height=16mm,
				align=center,
				inner sep=4pt,
				text=APCNavy
			},
			apirblock/.style={
				draw=APCTeal,
				fill=APCTealFill,
				line width=0.9pt,
				rounded corners=2.2pt,
				minimum width=39mm,
				minimum height=16mm,
				align=center,
				inner sep=4pt,
				text=APCNavy
			},
			plantblock/.style={
				draw=APCSlate,
				fill=APCGrayFill,
				line width=0.9pt,
				rounded corners=2.2pt,
				minimum width=33mm,
				minimum height=16mm,
				align=center,
				inner sep=4pt,
				text=APCNavy
			},
			sum/.style={
				draw=APCNavy,
				fill=white,
				line width=0.9pt,
				circle,
				minimum size=7.5mm,
				inner sep=0pt,
				text=APCNavy
			},
			group/.style={
				draw=APCBlue,
				fill=APCGroupFill,
				line width=0.8pt,
				rounded corners=3pt,
				inner xsep=5.5mm,
				inner ysep=6.0mm
			},
			siglabel/.style={
				fill=none,
				inner sep=1.2pt,
				text=APCNavy
			}
			]
			
			\node[sum] (sum) {$\sum$};
			\node[synthblock, right=14mm of sum] (synth)
			{Backstepping/BLF\\command synthesis};
			
			\node[apirblock, right=12mm of synth] (apir)
			{{\bfseries Admissibility-Preserving}\\
				{\bfseries Input Realization (APIR)}\\[0.4mm]
				$\dot{u}=\mathcal{G}(u)u_{c}-\mathcal{F}(u)$};
			
			\node[plantblock, right=15mm of apir] (plant)
			{Strict-feedback\\nonlinear plant};
			
			\begin{scope}[on background layer]
				\node[group, fit=(synth)(apir)] (controller) {};
			\end{scope}
			
			\node[
			anchor=south,
			text=APCBlue,
			font=\bfseries,
			fill=none,
			inner xsep=3pt,
			inner ysep=1pt
			] at ($(controller.north)+(0,1.4mm)$)
			{Admissibility-Preserving Controller (APC)};
			
			\coordinate (refin) at ($(sum.west)+(-18mm,0)$);
			\draw[signal] (refin) -- node[siglabel, above=1.1mm] {$y_{d}$} (sum.west);
			\draw[signal] (sum.east) -- node[siglabel, above=1.1mm] {$e$} (synth.west);
			\draw[signal] (synth.east) -- node[siglabel, above=1.1mm] {$u_{c}$} (apir.west);
			\draw[signal] (apir.east) -- node[siglabel, above=1.1mm] {$u\in\mathcal{U}$} (plant.west);
			
			\coordinate (ybranch) at ($(plant.east)+(12mm,0)$);
			\coordinate (yout) at ($(ybranch)+(13mm,0)$);

			\draw[signal,-] (plant.east) -- node[siglabel, above=1.1mm] {$y$} (ybranch);
			\fill[APCNavy] (ybranch) circle (1.25pt);
			\draw[signal] (ybranch) -- (yout);

			\coordinate (fbRight) at ($(ybranch)+(0,-20mm)$);
			\coordinate (fbLeft)  at ($(sum.south)+(0,-16.5mm)$);
			
			\draw[feedback]
			(ybranch) |- (fbRight)
			-- node[
			midway,
			below=1.1mm,
			text=APCFeedback,
			fill=none,
			inner xsep=2pt,
			inner ysep=1pt
			] {state/output feedback}
			(fbLeft) -| (sum.south);
			
			\node[font=\scriptsize, text=APCNavy, anchor=south]
			at ($(sum.west)+(-1.8mm,1.2mm)$) {$+$};
			
			\node[font=\scriptsize, text=APCFeedback, anchor=east]
			at ($(sum.south)+(-1.1mm,-0.3mm)$) {$-$};
			
		\end{tikzpicture}%
	}
	\caption{Canonical APC architecture. The control-synthesis layer generates $u_c$, while the APIR dynamically realizes the physical plant input $u\in\mathbb{U}$.}
	\label{fig:saturation_model_block}
\end{figure}
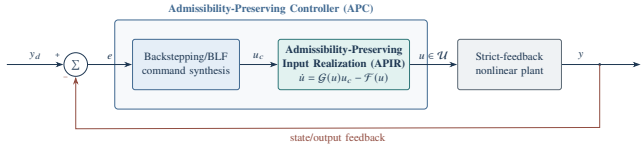

Observe that the proposed APIR naturally accommodates both symmetric and asymmetric actuator constraints through appropriate choices of $u_{\min}$ and $u_{\max}$. In particular, selecting $\lvert u_{\min}\rvert=u_{\max}$ recovers the symmetric case, whereas unequal bounds directly yield an asymmetric actuator realization without altering the realization dynamics.
\begin{remark}\label{rem:continuity}
	Although \eqref{eq:single_APIR} uses the branch variable $\rho$, the two vector-field branches agree at $u=0$, where $\lim_{u\to0^-}\dot{u}=\lim_{u\to0^+}\dot{u}=p_1u_c$. For an even integer $\gamma\geq2$, the vector field is locally Lipschitz in $u$ and is at least $\mathcal{C}^{\gamma-1}$ at the origin. When $u_{\max}\neq\lvert u_{\min}\rvert$ and $u_c(t)\neq0$, it is generically not $\mathcal{C}^{\gamma}$ at that instant. Thus, \emph{smooth} in the sequel means continuously differentiable with a user-selected finite regularity sufficient for the recursive design; the symmetric case is $\mathcal{C}^{\infty}$.
\end{remark}
\begin{theorem}[Intrinsic admissibility and compact interiority of APIR]\label{thm:APIR_admissibility}
	Consider \eqref{eq:single_APIR} with a piecewise-continuous command satisfying $\lvert u_c(t) \rvert \leq\xi$ for all $t\geq0$. Let $u_\xi^+\in(0,u_{\max})$ and $u_\xi^-\in(u_{\min},0)$ denote the unique roots of
	\begin{align}
		\xi\left[1-\left(\frac{u_\xi^+}{u_{\max}}\right)^\gamma\right]-p_2u_\xi^+ =&~0,\label{eq:exact_upper_root}\\
		-\xi\left[1-\left(\frac{u_\xi^-}{u_{\min}}\right)^\gamma\right]-p_2u_\xi^- =&~0.\label{eq:exact_lower_root}
	\end{align}
	For every $u(0)\in\mathbb{U}$, define
	\begin{equation}\label{eq:compact_APIR_interval}
		\mathbb{U}_{\xi,0}\coloneqq\left[\min\{u(0),u_\xi^-\},\max\{u(0),u_\xi^+\}\right]\subset\mathbb{U}.
	\end{equation}
	Then the solution exists for all $t\geq0$ and satisfies $u(t)\in\mathbb{U}_{\xi,0}$. Consequently, $\mathbb{U}$ is forward invariant and
	\begin{equation}\label{eq:G_positive_lower_bound}
		\underline{\mathcal{G}}_{\xi,0}\coloneqq\min_{u\in\mathbb{U}_{\xi,0}}\mathcal{G}(u)>0.
	\end{equation}
\end{theorem}
\begin{proof}
	For $u\in[0,u_{\max}]$, define $\Phi_+(u)=\xi[1-(u/u_{\max})^\gamma]-p_2u$. The function is strictly decreasing, with $\Phi_+(0)=\xi>0$ and $\Phi_+(u_{\max})=-p_2u_{\max}<0$, so \eqref{eq:exact_upper_root} has a unique solution. At any $u\geq u_\xi^+$,
	\begin{equation*}
		\dot{u}\leq p_1\Phi_+(u)\leq0.
	\end{equation*}
	Likewise, for $u\in[u_{\min},0]$, define $\Phi_-(u)=-\xi[1-(u/u_{\min})^\gamma]-p_2u$. It is strictly decreasing, satisfies $\Phi_-(u_{\min})=-p_2u_{\min}>0$ and $\Phi_-(0)=-\xi<0$, and therefore admits the unique root $u_\xi^-$. At any $u\leq u_\xi^-$,
	\begin{equation*}
		\dot{u}\geq p_1\Phi_-(u)\geq0.
	\end{equation*}
	Hence, the vector field points into the compact interval \eqref{eq:compact_APIR_interval} at both endpoints. Positive invariance follows, and boundedness of the state together with local Lipschitzness yield forward completeness. Since the compact interval lies strictly inside $\mathbb{U}$ and $\mathcal{G}(u)>0$ on $\mathbb{U}$, the minimum in \eqref{eq:G_positive_lower_bound} exists and is positive.
\end{proof}
The compact-interiority statement is essential for closed-loop synthesis as it converts strict membership in an open actuator set into a uniform lower bound on the gain that appears in the denominator of the commanded-input law.
\subsubsection{Parametric Sensitivity of APIR}\label{subsec:param_sensitivity}
The effect of the design parameters in the proposed APIR can be understood directly from its mathematical structure. To facilitate the analysis, define a function $\mathcal{S}(u)$ as
\begin{equation}
	\mathcal{S}(u)\coloneqq
	\rho\left(1-\left(\frac{u}{u_{\max}}\right)^{\gamma}\right)
	+
	\left(1-\rho\right)
	\left(1-\left(\frac{u}{u_{\min}}\right)^{\gamma}\right),
	\label{eq:S_def}
\end{equation}
which satisfies $\mathcal{S}(0)=1$, $\mathcal{S}(u_{\max})=\mathcal{S}(u_{\min})=0$. Using \eqref{eq:S_def},  the APIR dynamics in \eqref{eq:single_APIR} can be expressed compactly as
\begin{equation}\label{eq:factored_model}
	\dot{u}=p_{1}\left[\mathcal{S}(u)u_{c}-p_{2}u\right].
\end{equation}
It immediately follows from \eqref{eq:factored_model}  that the parameters $p_1$, $p_2$, and $\gamma$ affect the dynamics through three distinct mechanisms, which are discussed as follows.

\textit{Role of $p_{1}$:} For a constant command $u_c$, the change of time variable $\tau=p_1t$ removes $p_1$ from the autonomous APIR equation. Thus, under constant-command testing, $p_1$ changes the transient time scale without changing the equilibrium. For a general time-varying closed-loop command, the rescaled input is $u_c(\tau/p_1)$. Hence, $p_1$ should be interpreted as a bandwidth parameter rather than as an exact pure time scaling for arbitrary inputs.

\textit{Role of $p_{2}$:} For a constant command, the equilibrium $u^*$ is the unique solution of
\begin{equation}
	\mathcal{S}(u^*)u_c=p_2u^*.
	\label{eq:equilibrium_condition}
\end{equation}
This relation provides an exact gain-design rule. For a selected positive command $\xi$ and desired equilibrium $u_d^+\in(0,u_{\max})$, one may choose
\begin{equation}\label{eq:p2_positive_design}
	p_2=\frac{\xi\left[1-(u_d^+/u_{\max})^\gamma\right]}{u_d^+}.
\end{equation}
For a negative command $-\xi$ and desired equilibrium $u_d^-\in(u_{\min},0)$, the analogous choice is
\begin{equation}\label{eq:p2_negative_design}
	p_2=\frac{-\xi\left[1-(u_d^-/u_{\min})^\gamma\right]}{u_d^-}.
\end{equation}
Since one common $p_2$ acts on both asymmetric branches, it may be selected using the more restrictive of the two desired margins. Increasing $p_2$ moves both equilibria farther from the actuator limits, whereas decreasing $p_2$ increases steady-state utilization. The closed-form quantities previously obtained from comparison inequalities are conservative interior thresholds, not exact equilibria for general $\gamma$.

\textit{Role of $\gamma$:} The exponent $\gamma$ determines how rapidly the effective gain $\mathcal{S}(u)$ decreases near the actuator limits. Larger values preserve near-linear behavior over a wider portion of the admissible interval and concentrate the rolloff closer to the boundary. The exponent also determines the finite differentiability order at $u=0$. Since $\gamma$ appears in \eqref{eq:equilibrium_condition}, it additionally affects the equilibrium utilization for fixed $u_c$ and $p_2$. It is therefore a geometric, regularity, and steady-state design parameter. This regularity is characterized next.
\begin{proposition}
	\label{prop:smoothness_order}
	Let $u_{\max}\neq \lvert u_{\min}\rvert$. The right-hand side of
	\eqref{eq:single_APIR} is at least $\mathcal{C}^{\gamma-1}$ at $u=0$.
	All derivatives up to order $\gamma-1$ are continuous across $u=0$, and,
	whenever $u_c(t)\neq0$, the $\gamma^{\mathrm{th}}$ derivative exhibits the
	nonzero jump
	\begin{equation}
		\Delta_{\gamma}
		=
		p_1u_c\gamma!
		\left(
		\frac{1}{u_{\min}^{\gamma}}
		-
		\frac{1}{u_{\max}^{\gamma}}
		\right).
		\label{eq:smoothness_jump}
	\end{equation}
\end{proposition}
\begin{proof}
	Denote the right-hand side of \eqref{eq:single_APIR} by $h(u)$. For $u>0$ ($\rho=1$), $h(u)$ is given by
	\begin{equation}
		h_{+}(u)=p_1u_c\left[1-\left(\frac{u}{u_{\max}}\right)^{\gamma}\right]-p_1p_2u
		=p_1u_c-\frac{p_1u_c}{u_{\max}^\gamma}u^\gamma-p_1p_2u.
		\label{eq:h_plus}
	\end{equation}
	For $u\leq0$ ($\rho=0$), $h(u)$ is
	\begin{equation}
		h_{-}(u)=p_1u_c\left[1-\left(\frac{u}{u_{\min}}\right)^{\gamma}\right]-p_1p_2u
		=p_1u_c-\frac{p_1u_c}{u_{\min}^\gamma}u^\gamma-p_1p_2u.
		\label{eq:h_minus}
	\end{equation}
	Differentiating \eqref{eq:h_plus} $k$ times with respect to $u$, where $k=1,\ldots,\gamma-1$, yields
	\begin{equation}
		h_{+}^{(k)}(u)
		=-\frac{p_1u_c\gamma!}{(\gamma-k)!u_{\max}^\gamma}u^{\gamma-k}
		-p_1p_2\delta_{k,1},
		\label{eq:hplus_kth}
	\end{equation}
	where $\delta_{k,1}=1$ if $k=1$ and $\delta_{k,1}=0$ otherwise. Since $1\leq k\leq\gamma-1$ implies $\gamma-k\geq1$, the power $u^{\gamma-k}\to0$ as $u\to0^+$. Therefore
	\begin{equation}
		\lim_{u\to0^+}h_+^{(k)}(u)=-p_1p_2\delta_{k,1}.
		\label{eq:hplus_limit}
	\end{equation}
	By following a similar procedure, one may obtain
	\begin{equation}
		\lim_{u\to0^-}h_-^{(k)}(u)=-p_1p_2\delta_{k,1}.
		\label{eq:hminus_limit}
	\end{equation}
	One may notice from \eqref{eq:hplus_limit} and \eqref{eq:hminus_limit} that $ \lim_{u\to0^+}h_+^{(k)}(u)= \lim_{u\to0^-}h_-^{(k)}(u)$ holds true for every $k=1,\ldots,\gamma-1$, the function $h$ is $C^{\gamma-1}$ at $u=0$.
	
	Next, we calculate the $\gamma\textsuperscript{th}$ derivative, where $k=\gamma$. At $k=\gamma$, the factor $u^{\gamma-k}=u^0=1$ no longer vanishes. The one-sided limits are
	\begin{align}
		\lim_{u\to0^+}h_+^{(\gamma)}(u)
		&=-\frac{p_1u_c\gamma!}{u_{\max}^\gamma}-p_1p_2\delta_{\gamma,1},
		\label{eq:hplus_gamma}\\
		\lim_{u\to0^-}h_-^{(\gamma)}(u)
		&=-\frac{p_1u_c\gamma!}{u_{\min}^\gamma}-p_1p_2\delta_{\gamma,1}.
		\label{eq:hminus_gamma}
	\end{align}
	As $\gamma = 2n$, we have $\delta_{\gamma,1}=0$, so the $-p_1p_2\delta_{\gamma,1}$ terms in both expressions vanish. Their difference is
	\begin{equation*}
		\Delta_\gamma
		=\lim_{u\to0^+}h_+^{(\gamma)}(u)-\lim_{u\to0^-}h_-^{(\gamma)}(u)
		=p_1u_c\gamma!\left(\frac{1}{u_{\min}^{\gamma}}-\frac{1}{u_{\max}^{\gamma}}\right).
	\end{equation*}
	Since $\gamma$ is even, $u_{\min}^\gamma=\lvert u_{\min}  \rvert^\gamma>0$.
	If $u_{\max}\neq\lvert u_{\min}  \rvert$ then $1/u_{\min}^\gamma\neq1/u_{\max}^\gamma$,
	so, whenever $u_c(t)\neq0$, $\Delta_\gamma\neq0$ and $h\notin \mathcal{C}^\gamma$ at $u=0$. If $u_c(t)=0$ at an isolated instant, the jump vanishes at that instant; the uniform regularity guaranteed independently of the command is therefore $\mathcal{C}^{\gamma-1}$.
\end{proof}
\begin{remark}[(Symmetric case)] \label{rem:symmetric_smoothness}
	When $\lvert u_{\min}  \rvert=u_{\max}$, the jump \eqref{eq:smoothness_jump} vanishes identically, since $\gamma$ is even implies $u_{\min}^{\gamma}=\lvert u_{\min}  \rvert^{\gamma}=u_{\max}^{\gamma}$, so $1/u_{\min}^{\gamma}-1/u_{\max}^{\gamma}=0$.  This means that the two branches agree through orders $\gamma-1$ as well as $\gamma$. Moreover, the model is in fact $\mathcal{C}^{\infty}$ at $u=0$ in the symmetric case as all derivatives of $h_+(u)$ and $h_-(u)$ evaluated at $u=0$ match to every order, since both branches are polynomials in $u$ with identical coefficients when $\lvert u_{\min}  \rvert=u_{\max}$.
\end{remark}
\begin{remark}
	\Cref{prop:smoothness_order} characterizes the regularity of the asymmetric APIR family. In particular, $\gamma=2$, $\gamma=4$, and $\gamma=8$ guarantee at least $\mathcal{C}^1$, $\mathcal{C}^3$, and $\mathcal{C}^7$ continuity, respectively, independently of the command. Thus, $\gamma$ simultaneously controls the APIR gain-rolloff region and the differentiability order of the realization.
\end{remark}

\subsubsection{Numerical Sensitivity Study of APIR}
 To validate the analytical properties of the asymmetric APIR, we subject \eqref{eq:single_APIR} to the constant command $u_c(t)=15$, chosen to exceed the actuator limits $u_{\max}=10$ and $u_{\min}=-7$. This operating condition drives the realization toward its upper admissible boundary and directly probes the confining behavior established in \Cref{thm:APIR_admissibility}.

The effect of the parameter $p_{1}$ is illustrated in \Cref{fig:sat_SpeedMultiplier_p1} for $p_{1}\in\{2,5,15\}$ while all other parameters are kept fixed. Consistent with the time-scaling
argument discussed in \Cref{subsec:param_sensitivity}, all three trajectories converge to the same steady-state value while exhibiting different transient responses. Increasing $p_{1}$ produces a faster transient response (rise and settling time decrease), whereas decreasing $p_{1}$ yields a slower response. This behavior confirms that $p_{1}$ serves as an effective APIR bandwidth parameter without altering the constant-command equilibrium.

\Cref{fig:sat_Attenuation_p2} illustrates the effect of varying $p_{2}\in\{0.01,0.06,0.15\}$. For the constant positive command used in this study, the steady-state value is the unique positive root of \eqref{eq:equilibrium_condition}. Decreasing $p_2$ weakens the restoring action and moves this root closer to $u_{\max}$, whereas increasing $p_2$ enlarges the steady-state boundary margin. The plotted steady-state values are therefore obtained from the exact algebraic root rather than from the conservative comparison thresholds in \Cref{thm:APIR_admissibility}. Conversely, \eqref{eq:p2_positive_design} permits $p_2$ to be selected directly when a desired positive equilibrium utilization is specified.

\Cref{fig:sat_ShapeParameter_gamma} illustrates the response for $\gamma\in\{2,4,8\}$.  Larger values of $\gamma$ preserve near-linear behavior over a wider portion of the admissible control range, since the gain $\mathcal{S}(u)$ remains close to unity for a larger interval before collapsing toward zero near the boundary. Consequently, the actuator tracks the commanded input almost linearly until the saturation region is approached, at which point a rapid yet smooth transition to the asymptotic limit occurs.  As $\gamma$ increases, this transition becomes sharper, consistent with the smoothness analysis in \Cref{prop:smoothness_order}. The $\mathcal{C}^{\gamma-1}$ regularity improves with $\gamma$, but the transition zone simultaneously narrows, concentrating the nonlinear rolloff over a progressively smaller fraction of the operating range.
\begin{figure*}[!ht]
	\centering
	\begin{subfigure}[t]{0.33\linewidth}
		\centering
		\includegraphics[width=\linewidth]{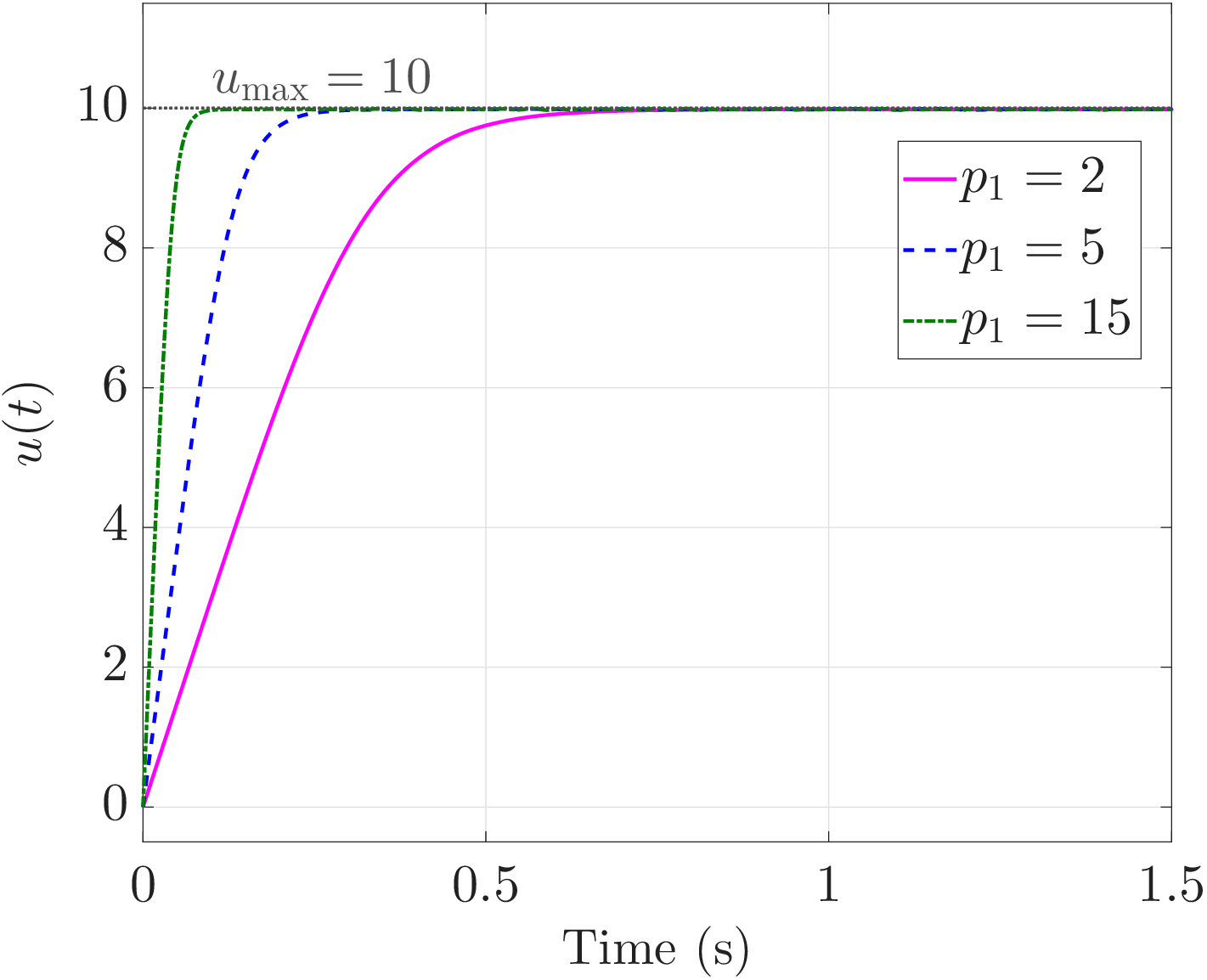}
		\caption{Varying $p_1\in\{2,5,15\}$ ($p_2=0.06$, $\gamma=2$).}
		\label{fig:sat_SpeedMultiplier_p1}
	\end{subfigure}%
	\begin{subfigure}[t]{0.33\linewidth}
		\centering
		\includegraphics[width=\linewidth]{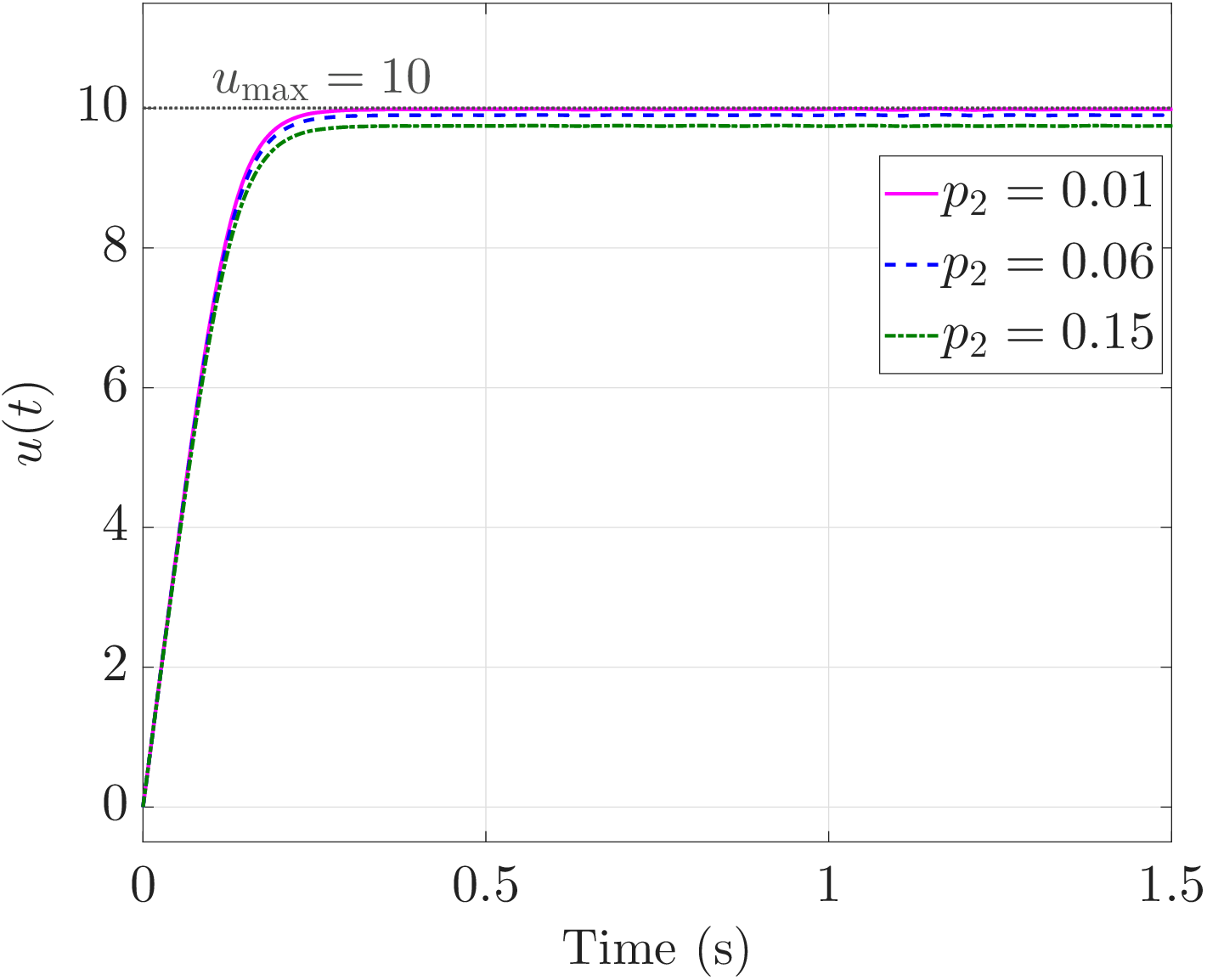}
		\caption{Varying $p_2\in\{0.01,0.06,0.15\}$ ($p_1=5$, $\gamma=2$).}
		\label{fig:sat_Attenuation_p2}
	\end{subfigure}%
	\begin{subfigure}[t]{0.33\linewidth}
		\centering
		\includegraphics[width=\linewidth]{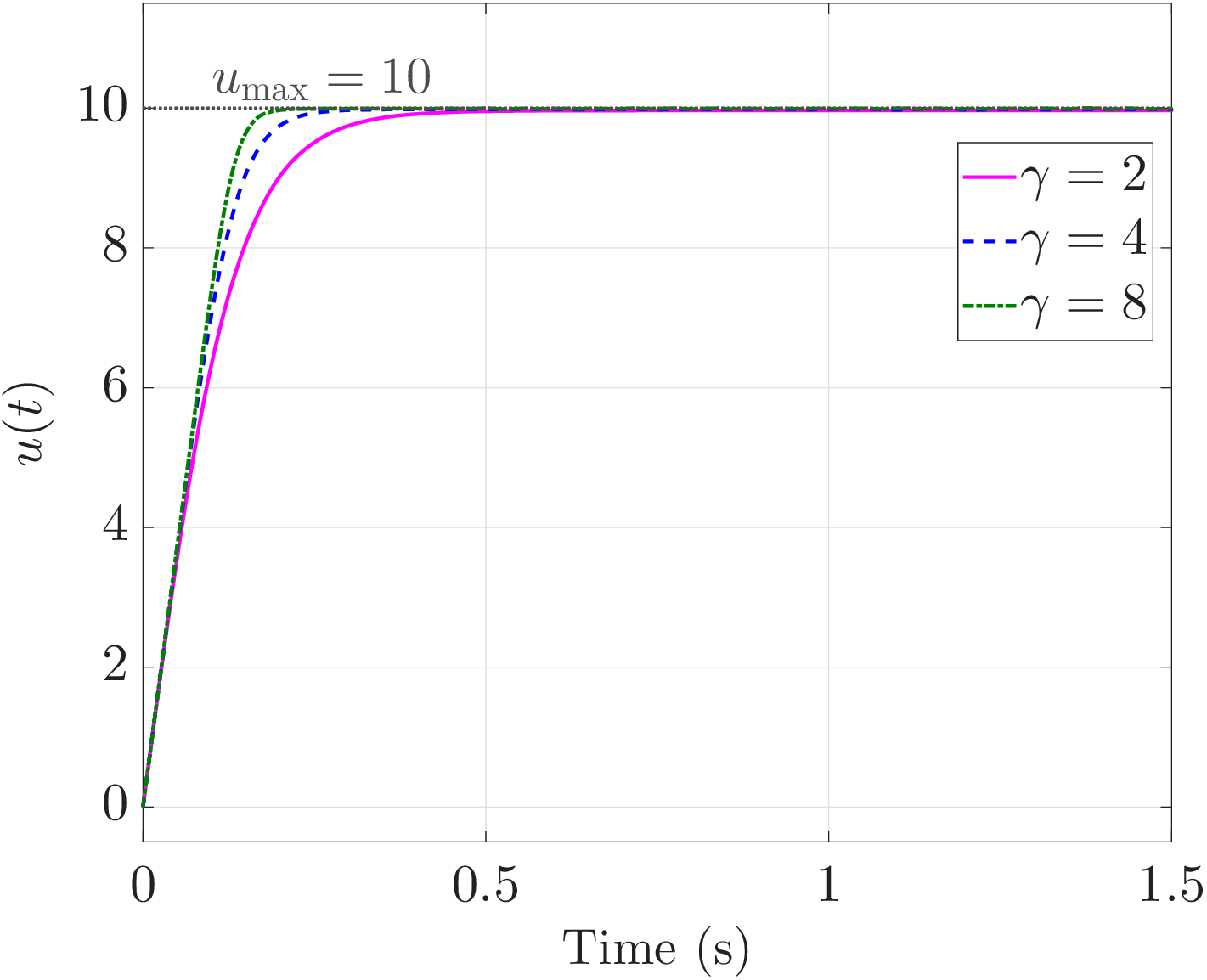}
		\caption{Varying $\gamma\in\{2,4,8\}$ ($p_1=5$, $p_2=0.06$).}
		\label{fig:sat_ShapeParameter_gamma}
	\end{subfigure}
	\caption{Parametric response of the asymmetric APIR~\eqref{eq:single_APIR} to the constant command $u_c=15$ with $u_{\max}=10$ and $u_{\min}=-7$.}
	\label{fig:actuator_response}
\end{figure*}
\begin{remark}[(Practical tuning guidelines)]
	A practical sequence is to select $\gamma$ from the desired regularity and rolloff profile, choose $p_2$ from \eqref{eq:p2_positive_design} or \eqref{eq:p2_negative_design} using the more restrictive asymmetric margin, and then select $p_1$ to obtain the desired constant-command response speed. As $\gamma$ also affects equilibrium utilization and a common $p_2$ couples the two asymmetric branches, the parameters are interpretable but not completely decoupled.
\end{remark}

\subsection{APIR-Augmented APC for Output Tracking}\label{subsec:output_tracking_bounded_input}
We now interconnect the APIR with the nonlinear plant to solve \Cref{problem1}. The augmented system is
\begin{subequations}\label{eq:system_dynamics_aug_global}
	\begin{align}
		\dot{x}_{i} =&~f_{i}\left(\bar{x}_{i}\right) + g_{i}\left(\bar{x}_{i}\right) x_{i+1}, \quad i = 1,2,\ldots,n-1 \label{eq:xi_dotaug_global}\\
		\dot{x}_{n} =&~f_{n}\left(\bar{x}_{n}\right) + g_{n}\left(\bar{x}_{n}\right) u , \label{eq:xn_dotaug_global}\\
		\dot{u} =&~\mathcal{G}(u)u_c-\mathcal{F}(u),\label{eq:u_dotaug_global}
	\end{align}
\end{subequations}
where $\mathcal{G}(u)\coloneqq p_1[\rho\{1-(u/u_{\max})^\gamma\} +(1-\rho)\{1-(u/u_{\min})^\gamma\}]$ and $\mathcal{F}(u)\coloneqq p_1 p_2 u$.
The APIR is interpreted as a controller-side dynamic realization whose output $u$ is supplied to the plant. The analysis therefore assumes that the lower-level actuator tracks this realized signal sufficiently accurately; unmodeled actuator dynamics or tracking error must be included explicitly before the mathematical input-safety guarantee can be transferred to hardware. Closed-loop admissibility additionally requires a self-consistent bound on $u_c$, because $\mathcal{G}(u)$ appears in the denominator of the command law.

The proposed APIR renders a continuously differentiable realized-input state that can be incorporated directly into recursive synthesis. We therefore use \emph{backstepping} to construct the APC law around this realized input. As in the usual output tracking problem with \emph{backstepping}, the following change of coordinates is made: $\varphi_{1} \coloneqq y - y_{d}$, $\varphi_{i} = x_{i} - \eta_{i-1}$ for $i=2,3,\cdots,n$, where $\eta_{i-1}$ is the stabilizing function for the $i\textsuperscript{th}$ step that needs to be designed, and $\varrho \coloneqq u - \eta_{n}$, where $\varrho$ is the actuator realization error introduced to handle the input saturation and $\eta_{n}$ is the stabilizing function for $n\textsuperscript{th}$ step. The additional state $\varrho$ naturally couples the recursive backstepping design with the admissibility-preserving actuator realization developed in the previous subsection. For brevity, only the first recursive step, the generic intermediate step, and the final two steps are presented, since the remaining derivations follow directly from the standard recursive backstepping procedure.

\textit{Step 1}: Define output tracking error as $\varphi_{1} \coloneqq y - y_{d}$. One can obtain the time derivative of $\varphi_1$ using \eqref{eq:xi_dot} as
\begin{equation} \label{eq:varphi1_dot}
	\dot{\varphi}_{1} = \dot{y}- \dot{y}_{d}= \dot{x}_{1}- \dot{y}_{d}= f_{1}(x_1) + g_1(x_{1}) x _{2}- \dot{y}_{d}.
\end{equation}
By viewing $x_{2}$ as a virtual control input, let us consider a Lyapunov function candidate $\mathcal{V}_{1}$ as $\mathcal{V}_{1} = \frac{1}{2}\varphi_{1}^2$, whose time derivative along the state trajectories results in $\dot{\mathcal{V}}_{1}=\varphi_{1}\dot{\varphi}_{1}$, which by substituting the value of $\dot{\varphi}_{1}$ from \eqref{eq:varphi1_dot}, becomes
\begin{equation} \label{eq:v1_dot_1}
	\dot{\mathcal{V}}_{1} = \varphi_{1}\left[ f_{1}(x_1) + g_1(x_{1}) x _{2}- \dot{y}_{d} \right].
\end{equation}
Now, define $\varphi_{2} \coloneqq x_{2} - \eta_{1}$, where $\eta_{1}$ is the stabilizing function. Using the relation $ x_{2} = \varphi_{2} +\eta_{1}$, the expression in \eqref{eq:v1_dot_1} becomes
\begin{equation} \label{eq:v1_dot_2}
	\dot{\mathcal{V}}_{1} = \varphi_{1}\left[ f_{1}(x_1) + g_1(x_{1}) \left(\varphi_{2} +\eta_{1} \right)- \dot{y}_{d} \right].
\end{equation}
If we design the stabilizing function $\eta_{1}$ as
\begin{equation} \label{eq:eta_1}
	\eta_{1} = \dfrac{1}{g_1(x_{1})}\left[ \dot{y}_{d} - f_{1}(x_1) - k_{1} \varphi_{1}\right],
\end{equation}
then the derivative of the Lyapunov function candidate, $\mathcal{V}_{1}$, becomes 
\begin{equation} \label{eq:v1_dot_final}
	\dot{\mathcal{V}}_{1} = -k_{1} \varphi_{1}^{2} +  g_1(x_{1}) \varphi_{1} \varphi_{2} .
\end{equation}

\textit{Step i ($2 \leq i \leq n-1$)}: In the $i\textsuperscript{th}$ step, we define the error variable as $\varphi_{i} = x_{i} - \eta_{i - 1}\, \forall \,i=2,\cdots,n-1$, which on differentiating with respect to time and using \eqref{eq:xi_dot} yields,
\begin{equation} \label{eq:varphi_i_dot}
	\dot{\varphi}_{i} = \dot{x}_{i} - \dot{\eta}_{i - 1} = f_{i}\left(\bar{x}_{i}\right) + g_{i}\left(\bar{x}_{i}\right) x_{i+1} - \dot{\eta}_{i - 1}.
\end{equation}
By considering $x_{i + 1}$ a virtual input to the $i\textsuperscript{th}$ subsystem, we choose a Lyapunov function candidate as $\mathcal{V}_{i} = \mathcal{V}_{i-1} + \frac{1}{2}\varphi_{i}^2$. By differentiating $\mathcal{V}_{i}$ with respect to time along the state trajectories, one has
\begin{equation*}
	\dot{\mathcal{V}}_{i} = \dot{\mathcal{V}}_{i-1} + \varphi_{i} \dot{\varphi}_{i},
\end{equation*}
which on using \eqref{eq:varphi_i_dot}, yields
\begin{equation} \label{eq:v_i_dot_1}
	\dot{\mathcal{V}}_{i} = \dot{\mathcal{V}}_{i-1} + \varphi_{i} \left[ f_{i}\left(\bar{x}_{i}\right) + g_{i}\left(\bar{x}_{i}\right) x_{i+1} - \dot{\eta}_{i - 1} \right],   
\end{equation}
Now define $\varphi_{i+1}\coloneqq x_{i+1} - \eta_{i}$, where $\eta_{i}$ is a stabilizing function. By substituting for $x_{i+1}$ into \eqref{eq:v_i_dot_1}, one may obtain
\begin{align} 
	\dot{\mathcal{V}}_{i} =&~ \varphi_{i} \left[ f_{i}\left(\bar{x}_{i}\right) + g_{i}\left(\bar{x}_{i}\right)\left( \varphi_{i+1} +\eta_{i} \right) - \dot{\eta}_{i - 1} \right]   \nonumber\\
	&~-\sum_{j=1}^{i-1} k_{j} \varphi_{j}^2+ g_{i-1} \left(\bar{x}_{i-1}\right) \varphi_{i-1}\varphi_{i}.\label{eq:v_i_dot_2} 
\end{align}
The stabilizing function $\eta_{i}$ is designed as
\begin{equation}
	\eta_i\coloneqq
	\frac{\dot{\eta}_{i-1}-f_i(\bar{x}_i)
		-g_{i-1}(\bar{x}_{i-1})\varphi_{i-1}-k_i\varphi_i}
	{g_i(\bar{x}_i)},\quad k_i>0,
	\label{eq:eta_i}
\end{equation}
which cancels $(g_{i-1}\varphi_{i-1}\varphi_i)$ from $\dot{\mathcal{V}}_{i-1}$ and adds damping $-k_i\varphi_i^2$, yielding
\begin{equation} \label{eq:v_i_dot_final}
	\dot{\mathcal{V}}_{i} = -\sum_{j=1}^{i} k_{j} \varphi_{j}^2 + g_{i}\left(\bar{x}_{i}\right)  \varphi_{i} \varphi_{i+1}.
\end{equation}

\textit{Step n}: In this step, the error variable is defined as $\varphi_{n} \coloneqq x_{n} - \eta_{n-1}$, where $\eta_{n-1}$ will be obtained from the $(n-1)\textsuperscript{th}$ step. Consider a Lyapunov function candidate as $\mathcal{V}_{n} = \mathcal{V}_{n-1} +\frac{1}{2}\varphi_{n}^2$, whose time derivative is obtained as $\dot{\mathcal{V}}_{n} = \dot{\mathcal{V}}_{n-1} +\varphi_{n} \dot{\varphi_{n}}$. Using \eqref{eq:v_i_dot_final}, one may write $\dot{\mathcal{V}}_{n}$ as
\begin{equation} \label{eq:v_n_dot_1}
	\dot{\mathcal{V}}_{n} = -\sum_{j=1}^{n-1} k_{j} \varphi_{j}^2 + g_{n-1}\left(\bar{x}_{n-1}\right) \varphi_{n-1}\varphi_{n} + \varphi_{n} \dot{\varphi}_{n}.
\end{equation}
The derivative of $\varphi_{n}$ can be obtained using \eqref{eq:xn_dot} as
\begin{equation} \label{eq:varphi_n_dot_1}
	\dot{\varphi_{n}} = \dot{x}_{n} - \dot{\eta}_{n-1} = f_{n}\left(\bar{x}_{n}\right) + g_{n}\left(\bar{x}_{n}\right) u - \dot{\eta}_{n-1}.
\end{equation}
Unlike conventional backstepping, the plant input is not synthesized directly. Instead, define the actuator realization error $\varrho\coloneqq u-\eta_n$. Using $u=\varrho+\eta_n$ and choosing
\begin{equation}
	\eta_n=
	\frac{\dot{\eta}_{n-1}-f_n(\bar{x}_n)-g_{n-1}(\bar{x}_{n-1})\varphi_{n-1}-k_n\varphi_n}
	{g_n(\bar{x}_n)},\quad k_n>0,
	\label{eq:eta_n}
\end{equation}
gives
\begin{equation}\label{eq:v_n_dot_final}
	\dot{\mathcal{V}}_n=-\sum_{j=1}^{n}k_j\varphi_j^2+g_n(\bar{x}_n)\varphi_n\varrho.
\end{equation}
\textit{Step n+1}: In this step, we focus on synthesizing the commanded input $u_c$ so that $\varrho$ also converges asymptotically to zero while preserving actuator admissibility.  Since we have $\varrho = u - \eta_{n}$, its time derivative can be obtained as $\dot{\varrho} = \dot{u} - \dot{\eta}_{n}$, which on substituting for $\dot{u}$ from \eqref{eq:single_APIR} results in 
\begin{align}
	\dot{\varrho} =&~ \mathcal{G}(u)u_c-\mathcal{F}(u) - \dot{\eta}_{n}. \label{eq:dot_varrho}
\end{align}
Consider a Lyapunov function candidate as $\mathcal{V}_{n+1} = \mathcal{V}_{n} +\frac{1}{2}\varrho^2$, whose time derivative along the state trajectory is given by $\dot{\mathcal{V}}_{n+1} = \dot{\mathcal{V}}_{n} + \varrho \dot{\varrho}$, which on using the relation in \eqref{eq:dot_varrho} and \eqref{eq:v_n_dot_final} yields
\begin{align}
	\dot{\mathcal{V}}_{n+1} &= -\sum_{j=1}^{n} k_{j} \varphi_{j}^2 + \varphi_{n}g_{n}\left(\bar{x}_{n}\right) \varrho + \varrho \left( \dot{u} - \dot{\eta}_{n} \right), \nonumber \\
	&=-\sum_{j=1}^{n} k_{j} \varphi_{j}^2 + \varphi_{n}g_{n}\left(\bar{x}_{n}\right) \varrho + \varrho \left( \mathcal{G}(u)u_c-\mathcal{F}(u) - \dot{\eta}_{n} \right). \label{eq:v_n+1_dot_1}
\end{align}
Now, we design the commanded control input $u_{c}$ as
\begin{equation} \label{eq:u_c}
	u_c=\frac{\mathcal{F}(u)+\dot{\eta}_n -g_n(\bar{x}_n)\varphi_n-k_{n+1}\varrho}{\mathcal{G}(u)},\quad k_{n+1}>0.
\end{equation}
With the commanded control $u_{c}$ given in \eqref{eq:u_c}, the derivative of Lyapunov function candidate given in \eqref{eq:v_n+1_dot_1} becomes
\begin{align*}
	\dot{\mathcal{V}}_{n+1} =&~-\sum_{j=1}^{n} k_{j} \varphi_{j}^2 + \varphi_{n}g_{n}\left(\bar{x}_{n}\right) \varrho + \varrho \left[ -\mathcal{F}(u) - \dot{\eta}_{n}   \right.\\
	&~\left.+  \frac{\mathcal{G}(u) \left (\mathcal{F}(u)+\dot{\eta}_n -g_n(\bar{x}_n)\varphi_n-k_{n+1}\varrho \right)}{\mathcal{G}(u)}\right]
\end{align*}
which, after some algebraic simplifications, results in
\begin{equation} \label{eq:v_n+1_dot_final}
	\dot{\mathcal{V}}_{n+1} = -\sum_{j=1}^{n} k_{j} \varphi_{j}^2 - k_{n+1}\varrho^2,
\end{equation}
which is negative definite on every domain on which the commanded-input map is well defined.

To close the apparent dependence between APIR admissibility and boundedness of the command, define
\begin{equation}\label{eq:command_numerator}
	\mathcal{N}(\mathbf{X},u,t)\coloneqq\mathcal{F}(u)+\dot{\eta}_n-g_n(\bar{x}_n)\varphi_n-k_{n+1}\varrho,
\end{equation}
where $\mathbf{X}=[\varphi_1,\ldots,\varphi_n,\varrho]^\top$. For $c>0$, let
$\Omega_c\coloneqq\{\mathbf{X}\mid\mathcal{V}_{n+1}(\mathbf{X})\leq c\}$.
We call the level $c$ \emph{APC-compatible} when the desired closed-loop motion and the commanded-input realization can be supported by a compact interior portion of the available actuator authority. Formally, $\eta_n^*(t)\coloneqq\eta_n \rvert_{\mathbf{X}=\mathbf{0}}$ must remain in a compact interior subset of $\mathbb{U}$ and there must exist $\xi>0$ such that
\begin{equation}\label{eq:self_consistency}
	Q_c(\xi)\coloneqq
	\sup_{\substack{\mathbf{X}\in\Omega_c,\,u\in\mathbb{U}_{\xi,0}\\t\geq0}}
	\frac{\lvert \mathcal{N}(\mathbf{X},u,t) \rvert}{\mathcal{G}(u)}
	\leq \xi.
\end{equation}
The ratio is well defined because \Cref{thm:APIR_admissibility} gives
$\mathcal{G}(u)\geq\underline{\mathcal{G}}_{\xi,0}>0$ on
$\mathbb{U}_{\xi,0}$. A more conservative but sometimes easier sufficient test is
$M_c(\xi)\leq\xi\underline{\mathcal{G}}_{\xi,0}$, where
$M_c(\xi)$ is the supremum of $\lvert \mathcal{N}\rvert$ on the same set.
We refer to this self-consistency property as \emph{APC compatibility}. It is the feasibility principle of the APC architecture: \eqref{eq:self_consistency} couples the selected initial-error region and reference demand to finite actuator authority and APIR conditioning.
\begin{lemma}[Bounded recursive quantities on compatible sublevels]\label{lem:recursive_boundedness}
	Under \Cref{assum:desired_traj,assum:control_effectiveness}, all virtual controls and the time derivatives required by the recursion are uniformly bounded on every compact APC-compatible sublevel.
\end{lemma}
\begin{proof}
	The result follows by induction. On a compact sublevel, $\varphi_1$ and the bounded reference determine a compact set for $x_1$, so $\eta_1$ and its required derivatives are bounded by the assumed smoothness of $f_1$, $g_1$, and $y_d$. If the assertion holds through Step $i-1$, then $x_i=\varphi_i+\eta_{i-1}$ is bounded. Smoothness of $f_i$ and $g_i$, together with the uniform lower bound on $\lvert g_{i} \rvert$, implies boundedness of $\eta_i$ and its required derivatives. The induction reaches $\eta_n$ and the command numerator \eqref{eq:command_numerator}.
\end{proof}
\begin{proposition}[Local nonvacuity of APC compatibility]\label{prop:compatibility_nonvacuity}
	Suppose that $\eta_n^*(t)$ remains uniformly separated from the actuator boundaries and that there exists $\xi_0>0$ for which
	\begin{equation}\label{eq:zero_error_compatibility}
		\sup_{t\geq0}
		\frac{\lvert \mathcal{N}(\mathbf{0},\eta_n^*(t),t) \rvert}
		{\mathcal{G}(\eta_n^*(t))}
		<\xi_0.
	\end{equation}
	Then there exists $c^*>0$ such that every sufficiently small sublevel
	$\Omega_c$, $0<c\leq c^*$, is APC-compatible.
\end{proposition}
\begin{proof}
	Uniform interiority provides a compact neighborhood on which $\mathcal{G}$ has a positive lower bound. By \Cref{lem:recursive_boundedness}, the ratio in \eqref{eq:self_consistency} is continuous on that neighborhood. The strict inequality in \eqref{eq:zero_error_compatibility} therefore persists on a sufficiently small Lyapunov sublevel by uniform continuity.
\end{proof}
\begin{theorem}[Compatible regional tracking and actuator admissibility]\label{thm:global_tracking}
	Consider \eqref{eq:system_dynamics_aug_global} under \Cref{assum:desired_traj,assum:control_effectiveness}, with virtual controls \eqref{eq:eta_1}--\eqref{eq:eta_n} and command \eqref{eq:u_c}. Let $c$ be APC-compatible, $\mathcal{V}_{n+1}(0)\leq c$, and $u(0)\in\mathbb{U}$. Then the closed-loop solution exists for all $t\geq0$, $u(t)\in\mathbb{U}_{\xi,0}\subset\mathbb{U}$, all closed-loop signals are bounded, and
	\begin{equation*}
		\|\mathbf{X}(t)\|\leq\sqrt{2\mathcal{V}_{n+1}(0)}e^{-2\vartheta_1t},\quad
		\vartheta_1\coloneqq\min_{1\leq j\leq n+1}k_j.
	\end{equation*}
	Consequently, $y(t)-y_d(t)\to0$ and $\varrho(t)\to0$.
\end{theorem}
\begin{proof}
	Consider the maximal interval on which $|u_c(t)|\leq\xi$. On that interval, \Cref{thm:APIR_admissibility} gives $u(t)\in\mathbb{U}_{\xi,0}$ and $\mathcal{G}(u(t))\geq\underline{\mathcal{G}}_{\xi,0}$. Equation \eqref{eq:v_n+1_dot_final} implies
	\begin{equation*}
		\dot{\mathcal{V}}_{n+1}\leq-2\vartheta_1\mathcal{V}_{n+1},
	\end{equation*}
	so $\Omega_c$ is invariant and the stated exponential estimate follows. Using \eqref{eq:self_consistency},
	\begin{equation*}
		|u_c(t)|=\frac{\lvert \mathcal{N}(\mathbf{X},u,t) \rvert}{\mathcal{G}(u(t))}
		\leq Q_c(\xi)\leq\xi.
	\end{equation*}
	Thus, the assumed command bound is preserved at the endpoint of every such interval, and a standard continuation argument extends the solution for all time. Boundedness of the transformed errors, the reference derivatives, and the recursively defined virtual controls then implies boundedness of the plant states and all controller signals.
\end{proof}
\begin{corollary}[Transfer of actuator safety under bounded realization error]\label{cor:tightened_apir}
	Let the physical actuator input be $u_{\mathrm{phys}}=u+d_u$, where
	$\lvert d_u(t) \rvert\leq\bar d_u$. If the APIR is designed on the tightened interval
	\begin{equation*}
		\mathbb{U}_{\mathrm{tight}}
		\coloneqq
		(u_{\min}+\bar d_u,\;u_{\max}-\bar d_u),
	\end{equation*}
	and $u(t)\in\mathbb{U}_{\mathrm{tight}}$, then
	$u_{\mathrm{phys}}(t)\in\mathbb{U}$ for all $t\geq0$. This corollary transfers the input-safety guarantee only; robust tracking in the presence of actuator-model mismatch requires an additional perturbation analysis.
\end{corollary}

\subsection{APC under Input--Output Admissibility}
We next address \Cref{problem2}. In addition to actuator protection, safety-critical applications may impose field-of-view, altitude, attitude, displacement, or other output envelopes that must hold throughout the transient \cite{doi:10.2514/1.G007770,doi:10.2514/1.G009445}. We use an asymmetric time-varying barrier Lyapunov function so that constant and symmetric constraints are included as special cases.

Since the recursive controller construction presented in the previous subsection remains applicable, \Cref{assum:desired_traj,assum:control_effectiveness} continue to hold.
\begin{assumption}\label{assum:y_bound_asy}
	The functions $\underline{y}(t)$ and $\overline{y}(t)$ satisfy $\underline{y}(t)<\overline{y}(t)$ and are sufficiently differentiable for the recursive design. They and their derivatives through order $n+1$ are uniformly bounded.
\end{assumption}
\begin{assumption}\label{assum:yd_bound_asy}
	There exists $\varepsilon_y>0$ such that
	\begin{equation*}
		\underline{y}(t)+\varepsilon_y\leq y_d(t)\leq\overline{y}(t)-\varepsilon_y,\, t\geq0.
	\end{equation*}
	Thus, the reference is uniformly separated from both output boundaries.
\end{assumption}
Define
\begin{equation*}
	\alpha(t)\coloneqq y_d(t)-\underline{y}(t)>0,\,
	\beta(t)\coloneqq\overline{y}(t)-y_d(t)>0,
\end{equation*}
and introduce the smooth asymmetric logarithmic barrier coordinate
\begin{equation}\label{eq:smooth_barrier_coordinate}
	z_1\coloneqq
	\ln\!\left(
	\frac{\beta(t)\,[\alpha(t)+\varphi_1]}
	{\alpha(t)\,[\beta(t)-\varphi_1]}
	\right).
\end{equation}
The transformation is a smooth diffeomorphism from
$\varphi_1\in(-\alpha(t),\beta(t))$ to $z_1\in\mathbb{R}$, satisfies
$z_1=0$ if and only if $\varphi_1=0$, and diverges at either output boundary.
Define
\begin{align}
	q(\varphi_1,t)
	&\coloneqq
	\frac{\partial z_1}{\partial\varphi_1}
	=
	\frac{1}{\alpha+\varphi_1}
	+
	\frac{1}{\beta-\varphi_1}>0, \label{eq:q_barrier}\\
	\psi(\varphi_1,t)
	&\coloneqq
	\left.\frac{\partial z_1}{\partial t}\right|_{\varphi_1}
	=
	\dot{\alpha}\left(\frac{1}{\alpha+\varphi_1}-\frac{1}{\alpha}\right)
	+
	\dot{\beta}\left(\frac{1}{\beta}-\frac{1}{\beta-\varphi_1}\right).
	\label{eq:psi_barrier}
\end{align}
Both $q$ and $\psi$ are smooth throughout the admissible output set, including at $\varphi_1=0$. Choose
\begin{equation}\label{eq:eta_1_asy}
	\eta_1=
	\frac{
		\dot y_d-f_1(x_1)
		-\bigl[\psi(\varphi_1,t)+k_1z_1\bigr]/q(\varphi_1,t)}
	{g_1(x_1)},
	\qquad k_1>0.
\end{equation}
Then
\begin{equation}\label{eq:z1_dot}
	\dot z_1=-k_1z_1+q(\varphi_1,t)g_1(x_1)\varphi_2.
\end{equation}
With the barrier Lyapunov function
\begin{equation}\label{eq:w}
	\mathcal{W}_1=\frac{1}{2}z_1^2,
\end{equation}
one obtains
\begin{equation}\label{eq:w_1_dot_final}
	\dot{\mathcal{W}}_1
	=-k_1z_1^2+\chi_1\varphi_2,
	\,
	\chi_1\coloneqq z_1q(\varphi_1,t)g_1(x_1).
\end{equation}
For $i=2,\ldots,n-1$, define
$\chi_i\coloneqq g_i(\bar{x}_i)\varphi_i$. The remaining virtual controls are
\begin{equation}\label{eq:eta_i_asy}
	\eta_i=
	\frac{\dot{\eta}_{i-1}-f_i(\bar{x}_i)-\chi_{i-1}-k_i\varphi_i}
	{g_i(\bar{x}_i)},
	\, i=2,\ldots,n,
\end{equation}
where $\varrho=u-\eta_n$. With
\begin{equation*}
	\mathcal{W}_{n+1}
	=
	\frac{1}{2}z_1^2
	+\frac{1}{2}\sum_{j=2}^{n}\varphi_j^2
	+\frac{1}{2}\varrho^2
\end{equation*}
and the command \eqref{eq:u_c}, direct substitution gives
\begin{equation}\label{eq:w_n+1_dot_final}
	\dot{\mathcal{W}}_{n+1}
	=
	-k_1z_1^2
	-\sum_{j=2}^{n}k_j\varphi_j^2
	-k_{n+1}\varrho^2.
\end{equation}
Unlike a branch-selected asymmetric BLF, \eqref{eq:smooth_barrier_coordinate} remains differentiable when the tracking error crosses zero and therefore supports the repeated virtual-control differentiations required by backstepping.

For $c>0$, define the full-state barrier sublevel set
\begin{equation*}
	\Omega_c^Y\coloneqq\left\{(z_1,\varphi_2,\ldots,\varphi_n,\varrho)\mid \mathcal{W}_{n+1}\leq c\right\}.
\end{equation*}
The level $c$ is called output-APC-compatible if the direct-ratio condition \eqref{eq:self_consistency}, with $\Omega_c$ replaced by $\Omega_c^Y$, holds and the zero-error terminal virtual input remains in a compact interior subset of $\mathbb{U}$. Since a finite $z_1$ is equivalent to strict output interiority, a compact sublevel automatically defines a compatible full-state safe set.

\begin{theorem}[Compatible input-output admissibility]\label{thm:semi_global_tracking}
	Suppose \Cref{assum:desired_traj,assum:control_effectiveness,assum:y_bound_asy,assum:yd_bound_asy} hold. Let $c$ be output-APC-compatible, let the initial full error state belong to $\Omega_c^Y$, and let $u(0)\in\mathbb{U}$. Then the closed-loop solution exists for all $t\geq0$, all signals remain bounded,
	\begin{equation*}
		\underline{y}(t)<y(t)<\overline{y}(t),\quad
		u(t)\in\mathbb{U}_{\xi,0}\subset\mathbb{U},
	\end{equation*}
	and $y(t)-y_d(t)\to0$.
\end{theorem}
\begin{proof}
	Equation \eqref{eq:w_n+1_dot_final} gives
	\begin{equation*}
		\dot{\mathcal{W}}_{n+1}
		\leq
		-2\vartheta_2\mathcal{W}_{n+1},
		\qquad
		\vartheta_2\coloneqq\min\{k_1,k_2,\ldots,k_{n+1}\}>0.
	\end{equation*}
	Hence,
	$\mathcal{W}_{n+1}(t)\leq\mathcal{W}_{n+1}(0)e^{-2\vartheta_2t}$.
	A finite value of $\mathcal{W}_{n+1}$ implies a finite $z_1$, and the inverse of
	\eqref{eq:smooth_barrier_coordinate} therefore remains strictly inside
	$(-\alpha(t),\beta(t))$. Thus,
	$\underline y(t)<y(t)<\overline y(t)$ for all time.
	The same direct-ratio bootstrap argument used in \Cref{thm:global_tracking},
	now on $\Omega_c^Y$, yields the command bound, compact APIR interiority, and
	forward completeness. Finally, decay of $\mathcal{W}_{n+1}$ implies
	$z_1\to0$, $\varphi_i\to0$, and $\varrho\to0$; invertibility of the barrier
	coordinate gives $\varphi_1\to0$.
\end{proof}

\subsection{Cascaded APIR for Magnitude--Rate APC}\label{subsec:MRS_extension}
We extend the proposed APC framework to simultaneously enforce both magnitude and rate constraints on the plant input. In addition to requiring $u(t) \in \mathbb{U}$, we impose an asymmetric rate constraint as
\begin{equation}
	\dot{u}(t) \in \mathbb{S} \coloneqq \left\{ \dot{u} \in \mathbb{R} \;\middle|\; \dot{u}_{\min} < \dot{u}(t) < \dot{u}_{\max} \right\}, \quad \forall\, t \geq 0,
	\label{eq:rate_set}
\end{equation}
where $\dot{u}_{\min} < 0$ and $\dot{u}_{\max} > 0$ are prescribed lower and upper rate limits.

Within APC, simultaneous magnitude--rate admissibility is realized by cascading the same APIR primitive. Introducing the intermediate state $v \in \mathbb{R}$ gives the two-layer APIR
\begin{subequations}\label{eq:cascaded_APIR}
	\begin{align}
		\dot{u} &= \mathcal{G}(u)v - \mathcal{F}(u),
		\label{eq:MRS_u_layer} \\
		\dot{v} &= \mathcal{G}_v(v)U_c - \mathcal{F}_v(v),
		\label{eq:MRS_v_layer}
	\end{align}
\end{subequations}
where $U_c \in \mathbb{R}$ is the new commanded input to be designed, $\mathcal{G}(u)$ and $\mathcal{F}(u)$ are identical to those introduced in the APC model in \eqref{eq:system_dynamics_aug_global}, and the functions associated with the newly introduced rate layer are defined as
\begin{align*}
	\mathcal{G}_v(v) &\coloneqq q_1\left[\sigma\left\{1-\left(\frac{v}{v_{\max}}\right)^{\mu}\right\} + (1-\sigma)\left\{1-\left(\frac{v}{v_{\min}}\right)^{\mu}\right\}\right], \\
	\mathcal{F}_v(v) &\coloneqq q_1 q_2 v,
\end{align*}
with $\sigma = 1$ if $v > 0$, $\sigma = 0$ if $v \leq 0$, $q_1, q_2 \in \mathbb{R}_+$, $\mu = 2n$ with $n \in \mathbb{N}$ is an even integer, and $v_{\max} > 0$, $v_{\min} < 0$ being the rate-layer bounds whose selection will be discussed later.

	Note that \eqref{eq:MRS_u_layer} is structurally identical to the single-layer APIR proposed \eqref{eq:single_APIR}, with the intermediate state $v$ replacing the commanded input $u_c$. Equation \eqref{eq:MRS_v_layer} constitutes a second APIR that confines $v$ to its admissible interval $\mathbb{V} \coloneqq (v_{\min}, v_{\max})$. Since both layers share the same mathematical structure, the continuity property established in \Cref{rem:continuity}, the smoothness characterization of \Cref{prop:smoothness_order}, and the parametric sensitivity analysis of \Cref{subsec:param_sensitivity} apply to each layer independently, with the parameter triplet $(q_1, q_2, \mu)$ governing the rate layer in the same interpretable manner as $(p_1, p_2, \gamma)$ governs the magnitude layer.

	The symmetric rate saturation case $|v_{\min}| = v_{\max}$ can be considered as a special case of \eqref{eq:MRS_v_layer}. Thus, the proposed formulation accommodates asymmetric rate limits directly, which is practically relevant in systems where the rate of increase and the rate of decrease of the actuator output are subject to different physical constraints.

The following feasibility condition ensures that the prescribed rate limits are consistent with the magnitude layer APIR.
\begin{assumption}[Rate feasibility]\label{assum:rate_feasibility}
	The prescribed asymmetric rate limits satisfy
	\begin{equation}\label{eq:rate_feasibility}
		\dot{u}_{\max} > p_1 p_2 \lvert u_{\min} \rvert, \, \lvert \dot{u}_{\min}\rvert > p_1 p_2 u_{\max}.
	\end{equation}
\end{assumption}
\begin{remark}
	\Cref{assum:rate_feasibility} is a compatibility condition between the requested slew-rate authority and the restoring term of the magnitude APIR. Decreasing $p_2$ can enlarge the feasible rate range, but it also moves the magnitude equilibrium closer to the actuator boundary and can reduce the interior gain margin. The condition therefore exposes a design tradeoff rather than imposing a cost-free or automatically nonrestrictive requirement.
\end{remark}
Under \Cref{assum:rate_feasibility}, the admissible bounds of the intermediate state are selected as
\begin{equation}
	v_{\max} \coloneqq \frac{\dot{u}_{\max}}{p_1} - p_2\abs{u_{\min}}, \,
	v_{\min} \coloneqq -\frac{\abs{\dot{u}_{\min}}}{p_1} + p_2 u_{\max},
	\label{eq:rate_layer_bounds} 
\end{equation}
which immediately guarantees $v_{\max}>0$ and $v_{\min}<0$. Consequently, the rate-layer dynamics are well posed and the admissible set $\mathbb{V}$ is nonempty.
\begin{theorem}[Admissibility of the cascaded APIR]\label{thm:MRS_sat}
	Consider the cascaded APIR model \eqref{eq:cascaded_APIR}, where the intermediate-state bounds are selected according to \eqref{eq:rate_layer_bounds} under  \Cref{assum:rate_feasibility}. Suppose that the commanded input satisfies $|U_c(t)| \leq \Xi < \infty$ for all $t \geq 0$, with $u(0) \in \mathbb{U}$ and $v(0) \in \mathbb{V} \coloneqq (v_{\min}, v_{\max})$, then for all $t \geq 0$,  $u(t) \in \mathbb{U}$, $v(t) \in \mathbb{V}$, and $\dot{u}(t) \in \mathbb{S}$.
\end{theorem}
\begin{proof}
	We decompose the proof into three parts. We first prove that $v(t) \in \mathbb{V}$, subsequently we prove that the input magnitude and rate remain confined to their respective safe set.
	
	We start by showing the confinement of the intermediate state variable $v(t)$. Layer \eqref{eq:MRS_v_layer} is structurally identical to \eqref{eq:single_APIR} with the substitutions $u \mapsto v$, $p_1 \mapsto q_1$, $p_2 \mapsto q_2$, $\gamma \mapsto \mu$, $u_{\max} \mapsto v_{\max}$, $u_{\min} \mapsto v_{\min}$, and $u_c \mapsto U_c$. Since $ \lvert U_c(t)  \rvert \leq \Xi < \infty$ and $v(0) \in \mathbb{V}$, \Cref{thm:APIR_admissibility} applied to \eqref{eq:MRS_v_layer} yields $v(t) \in \mathbb{V}$ for all $t \geq 0$. 
	
	Next, we show that the input magnitude $u(t)\in \mathbb{U}$. As the auxiliary state $v(t)$ is confined to a compact set, it satisfies the uniform bound $|v(t)| \le V^* \coloneqq \max(|v_{\min}|, v_{\max}) < \infty$. The intermediate state $v(t)$ now acts as a bounded command input to the magnitude sub-layer \eqref{eq:MRS_u_layer}. Applying \Cref{thm:APIR_admissibility} to \eqref{eq:MRS_u_layer} with $v(t)$ in the role of the bounded commanded input $u_c$ establishes $u(t) \in \mathbb{U}$ for all $t \geq 0$.
	
	Finally, we show confinement of the input rate. The input rate is given by $\dot{u}(t)=\mathcal{G}(u)v - \mathcal{F}(u) = p_{1} [ \mathcal{S}(u) v - p_{2} u ]$. To demonstrate that $\dot{u}(t)$ is strictly bounded within the safe physical rate set $\mathbb{S} \coloneqq (\dot{u}_{\min}, \dot{u}_{\max})$, we derive the algebraic bounds of this relation. To find the upper bound of $\dot{u}(t)$, we maximize the right-hand side of \eqref{eq:MRS_u_layer}. We have $v(t) < v_{\max}$ and by definition, $\mathcal{S}(u) \in [0, 1]$. Therefore, the term $\mathcal{S}(u) v$ is strictly bounded above by $v_{\max}$.  We also have, $u(t) > u_{\min}$. Since $u_{\min} < 0$ and $p_{2} > 0$, the term $-p_{2} u$ is strictly bounded above by $-p_{2} u_{\min} = p_{2} \lvert u_{\min}  \rvert > 0$. Using the definition of $v_{\max}$ from \eqref{eq:rate_layer_bounds}, one may obtain the upper rate bound as
	\begin{align}
		\dot{u} &= \mathcal{G}(u)v - \mathcal{F}(u) < p_1 v_{\max} + p_1 p_2\lvert u_{\min}  \rvert, \nonumber\\
		&< p_1\left(\frac{\dot{u}_{\max}}{p_1} - p_2\lvert u_{\min}  \rvert\right) + p_1 p_2\lvert u_{\min}  \rvert = \dot{u}_{\max}.
	\end{align}
	To find the lower bound of $\dot{u}(t)$, we minimize the right-hand side of \eqref{eq:MRS_u_layer}. The term $\mathcal{S}(u) v$ is strictly bounded below by $v_{\min}$. As $u(t) < u_{\max}$, the term $-p_{2} u$ is strictly bounded below by $-p_{2} u_{\max}$. Applying the definition of $v_{\min}$ from \eqref{eq:rate_layer_bounds}, we get the lower bound as
	\begin{align}
		\dot{u} &= \mathcal{G}(u)v - \mathcal{F}(u) > p_1 v_{\min} - p_1 p_2 u_{\max}, \nonumber\\
		&> p_1\left(-\frac{\abs{\dot{u}_{\min}}}{p_1} + p_2 u_{\max}\right) - p_1 p_2 u_{\max} = \dot{u}_{\min}.
	\end{align}
	Therefore, $\dot{u}(t) \in (\dot{u}_{\min}, \dot{u}_{\max}) = \mathbb{S}$ for all $t \geq 0$. This concludes the proof.
\end{proof}
Having established that the proposed APIR model guarantees simultaneous satisfaction of the actuator magnitude and rate constraints, we now integrate the cascaded saturation dynamics with the strict-feedback nonlinear plant to address the output tracking problem. The objective is to design a smooth commanded input $U_c$ such that the plant output asymptotically tracks a desired reference while simultaneously satisfying $u(t)\in\mathbb U$ and $\dot u(t)\in\mathbb S$, for all time. Owing to the smoothness of both saturation layers, the recursive backstepping framework developed in the previous subsection can be naturally extended by introducing one additional design step associated with the rate layer. The augmented system, integrating the plant with the asymmetric MRS model, is expressed as
\begin{subequations}\label{eq:system_aug_MRS}
	\begin{align}
		\dot{x}_i &= f_i(\bar{x}_i) + g_i(\bar{x}_i) x_{i+1}, \quad i = 1, \ldots, n-1, \\
		\dot{x}_n &= f_n(\bar{x}_n) + g_n(\bar{x}_n) u, \\
		\dot{u} &= \mathcal{G}(u) v - \mathcal{F}(u), \label{eq:aug_u_MRS}\\
		\dot{v} &= \mathcal{G}_v(v) U_c - \mathcal{F}_v(v). \label{eq:aug_v_MRS}
	\end{align}
\end{subequations}
The coordinate transformation is defined as $\varphi_{1} \coloneqq y - y_{d}$, $\varphi_{i} \coloneqq x_{i} - \eta_{i-1}$ for $i=2,\ldots,n$, $\varrho \coloneqq u - \eta_{n}$ for the magnitude-state realization error, and $\varpi \coloneqq v - \eta_{n+1}$ for the rate state. Here, $\eta_{i}$ represents the virtual stabilizing functions designed at each step. Since the first $n$ design steps are identical to those developed for the APC problem, only the additional recursive steps associated with the APIR extension are presented in detail below.

\textit{Steps 1 to $n$}: The first $n$ steps of the backstepping procedure remain identical to the design discussed in \Cref{subsec:output_tracking_bounded_input}. Following the derivation up to Step $n$, the stabilizing function $\eta_n$ is chosen as
\begin{equation}
	\eta_n = \frac{\dot{\eta}_{n-1}-f_n(\bar{x}_n) -g_{n-1}(\bar{x}_{n-1})\varphi_{n-1}-k_n\varphi_n}{g_n(\bar{x}_n)}, \quad k_n>0. \label{eq:eta_n_mrs}
\end{equation}
Utilizing the Lyapunov function candidate $\mathcal{V}_{n} = \frac{1}{2}\sum_{j=1}^{n}\varphi_{j}^{2}$, the time derivative along the state trajectories evaluated with the substitution $u = \varrho + \eta_n$ yields
\begin{equation} \label{eq:v_n_dot_final_mrs}
	\dot{\mathcal{V}}_{n} = -\sum_{j=1}^{n} k_{j} \varphi_{j}^2 + \varphi_{n}g_{n}\left(\bar{x}_{n}\right) \varrho.
\end{equation}

\textit{Step $n+1$:} Define the error variable $\varrho \coloneqq u - \eta_n$. Differentiating $\varrho$ with respect to time and substituting \eqref{eq:MRS_u_layer} gives
\begin{align}
	\dot{\varrho} &= \dot{u} - \dot{\eta}_{n} = \mathcal{G}(u)v - \mathcal{F}(u) - \dot{\eta}_{n}. \label{eq:dot_varrho1}
\end{align}
Consider the augmented Lyapunov function candidate $\mathcal{V}_{n+1} = \mathcal{V}_{n} + \frac{1}{2}\varrho^2$. Taking the time derivative and substituting \eqref{eq:v_n_dot_final_mrs} and \eqref{eq:dot_varrho1}, one obtains
\begin{align}
	\dot{\mathcal{V}}_{n+1} &= -\sum_{j=1}^{n} k_j \varphi_j^2 + g_n(\bar{x}_n)\varphi_n\varrho
	+ \varrho\left[\mathcal{G}(u)v - \mathcal{F}(u) - \dot{\eta}_n\right].
	\label{eq:Vn1_dot_MRS}
\end{align}
We now view $v$ as the virtual control input at this step. Define $\varpi \coloneqq v - \eta_{n+1}$, where $\eta_{n+1}$ is a stabilizing function to be designed. Substituting $v = \varpi + \eta_{n+1}$ into \eqref{eq:Vn1_dot_MRS} and designing
\begin{equation}
	\eta_{n+1} = \frac{\mathcal{F}(u) + \dot{\eta}_n - g_n(\bar{x}_n)\varphi_n - k_{n+1}\varrho}{\mathcal{G}(u)}, \quad k_{n+1} > 0,
	\label{eq:eta_n1_MRS}
\end{equation}
yields
\begin{equation}  \label{eq:Vn1_dot_final_MRS}
	\dot{\mathcal{V}}_{n+1} = -\sum_{j=1}^{n} k_j \varphi_j^2 - k_{n+1}\varrho^2 + \mathcal{G}(u)\varrho\varpi.
\end{equation}

\textit{Step $n+2$:} In this final step, we design the actual commanded input $U_c$. The error variable is $\varpi \coloneqq v - \eta_{n+1}$. The time derivative of $\varpi$ can be obtained by using \eqref{eq:MRS_v_layer}, as
\begin{equation}
	\dot{\varpi} = \dot{v} - \dot{\eta}_{n+1} =  \mathcal{G}_v(v)U_c - \mathcal{F}_v(v) - \dot{\eta}_{n+1}. \label{eq:dot_varrho2}
\end{equation}
Consider the final Lyapunov function candidate $\mathcal{V}_{n+2} = \mathcal{V}_{n+1} + \frac{1}{2}\varpi^2$, whose time derivative is given by
\begin{equation}
	\dot{\mathcal{V}}_{n+2} = \dot{\mathcal{V}}_{n+1} + \varpi \dot{\varpi}. \label{eq:v_n+2_dot_1_mrs}
\end{equation}
Substituting \eqref{eq:Vn1_dot_final_MRS} and \eqref{eq:dot_varrho2} into \eqref{eq:v_n+2_dot_1_mrs} yields
\begin{align}
	\dot{\mathcal{V}}_{n+2} &= -\sum_{j=1}^{n} k_j\varphi_j^2 - k_{n+1}\varrho^2 + \varpi\left[\mathcal{G}_v(v)U_c - \mathcal{F}_v(v) - \dot{\eta}_{n+1}\right] \nonumber \\
	&\quad + \mathcal{G}(u)\varrho\varpi.
	\label{eq:Vn2_dot_MRS}
\end{align}
We design the commanded input $U_c$ as
\begin{equation}
	U_c = \frac{\mathcal{F}_v(v) + \dot{\eta}_{n+1} - \mathcal{G}(u)\varrho - k_{n+2}\varpi}{\mathcal{G}_v(v)}, \quad k_{n+2} > 0,
	\label{eq:U_c_MRS}
\end{equation}
Substituting \eqref{eq:U_c_MRS} into \eqref{eq:Vn2_dot_MRS} and simplifying yields
\begin{equation}
	\dot{\mathcal{V}}_{n+2} = -\sum_{j=1}^{n} k_j\varphi_j^2 - k_{n+1}\varrho^2 - k_{n+2}\varpi^2,
	\label{eq:Vn2_dot_final_MRS}
\end{equation}
which is negative definite on a compatible domain where both realization gains are uniformly positive.
For the closed-loop two-layer design, let
\begin{align}
	\mathcal{N}_u
	&\coloneqq
	\mathcal{F}(u)+\dot{\eta}_n-g_n(\bar{x}_n)\varphi_n-k_{n+1}\varrho,
	\label{eq:Nu_two_layer}\\
	\mathcal{N}_v
	&\coloneqq
	\mathcal{F}_v(v)+\dot{\eta}_{n+1}-\mathcal{G}(u)\varrho-k_{n+2}\varpi.
	\label{eq:Nv_two_layer}
\end{align}
For a sublevel $\Omega_c^{(2)}$ of $\mathcal{V}_{n+2}$, call $c$
\emph{two-layer compatible} if the desired virtual inputs remain in compact
interior subsets of $\mathbb{U}$ and $\mathbb{V}$ and there exist
$\xi_u,\xi_v>0$ such that
\begin{align}
	\sup_{\Omega_c^{(2)}\times\mathbb{U}_{\xi_u,0}}
	\frac{|\mathcal{N}_u|}{\mathcal{G}(u)}
	&\leq\xi_u,
	\label{eq:two_layer_compat_u}\\
	\sup_{\Omega_c^{(2)}\times\mathbb{U}_{\xi_u,0}\times\mathbb{V}_{\xi_v,0}}
	\frac{|\mathcal{N}_v|}{\mathcal{G}_v(v)}
	&\leq\xi_v.
	\label{eq:two_layer_compat_v}
\end{align}
The suprema also extend over $t\geq0$. These inequalities provide explicit
positive lower bounds for both realization gains and close the continuation
argument at both layers.
\begin{theorem}[Compatible tracking under cascaded APIR]\label{thm:MRS_tracking}
	Consider \eqref{eq:system_aug_MRS} under \Cref{assum:desired_traj,assum:control_effectiveness,assum:rate_feasibility}. If the initial error state lies in a two-layer-compatible Lyapunov sublevel, $u(0)\in\mathbb{U}$, and $v(0)\in\mathbb{V}$, then the closed-loop solution is forward complete, all transformed errors converge to zero, all signals remain bounded, and $u(t)\in\mathbb{U},~ v(t)\in\mathbb{V},~\dot{u}(t)\in\mathbb{S},~ t\geq0$.
\end{theorem}
\begin{proof}
	Equation \eqref{eq:Vn2_dot_final_MRS} gives exponential decay of the augmented Lyapunov function on the selected sublevel. The two self-consistency inequalities provide positive lower bounds for $\mathcal{G}(u)$ and $\mathcal{G}_v(v)$ and close the command-boundedness continuation argument at both layers. The module result in \Cref{thm:MRS_sat} then yields magnitude, intermediate-state, and rate admissibility. Boundedness of the remaining signals follows recursively from the virtual-control definitions.
\end{proof}
	Compared with the single-layer APIR, the two-layer construction introduces one additional backstepping step and the parameter triplet $(q_1,q_2,\mu)$. The interpretations in \Cref{subsec:param_sensitivity} carry over to the rate layer: $q_1$ controls the constant-command response time, $q_2$ affects the equilibrium margin from $v_{\min}$ and $v_{\max}$, and $\mu$ governs rolloff, regularity, and equilibrium utilization. A practical sequential design first selects the magnitude-layer parameters to satisfy the actuator-magnitude and rate-feasibility conditions, and then selects the rate-layer parameters. The tuning is interpretable, but the two layers are coupled through the feasible interval \eqref{eq:rate_layer_bounds} and through the closed-loop compatibility inequalities.
	
\subsubsection{General $r$-Layer Cascaded APIR}
\label{subsec:general_r_APIR}
The two-layer APIR enforces the physical magnitude constraint $u\in\mathbb{U}$ and rate constraint $\dot{u}\in\mathbb{S}$. The following $r$-layer construction generalizes the realization architecture by constraining every internal APIR state and every layer-output rate. These statements do not, by themselves, identify $w_k$ with $u^{(k)}$ for $k\geq2$. Such a physical higher-derivative interpretation requires an explicit recursive calculation of derivatives.  

To this end, let $w_0 \coloneqq u$ and introduce $r$ auxiliary APIR states $w_1, w_2, \ldots, w_r \in \mathbb{R}$. For each $k = 0, 1, \ldots, r$, define 
\begin{equation*}
	\mathcal{S}_k(w_k)\coloneqq \sigma_k\left[1-\left(\frac{w_k}{w_{k,\max}}\right)^{\gamma_k}\right] + (1-\sigma_k)\left[1-\left(\frac{w_k}{w_{k,\min}}\right)^{\gamma_k}\right],
\end{equation*}
where $\sigma_k = 1$ if $w_k > 0$ and $\sigma_k = 0$ otherwise,
$\gamma_k = 2n_k$ with $n_k \in \mathbb{N}$ is an even integer, and
$w_{k,\min} < 0 < w_{k,\max}$ are the prescribed bounds for the $k$-th
layer. The corresponding realization gain and restoring functions are
\begin{equation}\label{eq:GkFk_def}
	\mathcal{G}_k(w_k) \coloneqq p_{1,k}\mathcal{S}_k(w_k), 
	\mathcal{F}_k(w_k) \coloneqq p_{1,k}p_{2,k}w_k,
\end{equation}
with $p_{1,k}, p_{2,k} \in \mathbb{R}_{+}$. The admissible set for each
layer state is denoted $\mathbb{W}_k \coloneqq (w_{k,\min}, w_{k,\max})$,
with $\mathbb{W}_0 = \mathbb{U}$. The $r$-layer cascaded APIR is then
defined as
\begin{subequations}\label{eq:r_layer_APIR}
	\begin{align}
		\dot{w}_{k-1}&= \mathcal{G}_{k-1}(w_{k-1})w_k - \mathcal{F}_{k-1}(w_{k-1}), k = 1, \ldots, r, \label{eq:r_layer_APIR_inner} \\
		\dot{w}_r
		&= \mathcal{G}_r(w_r)U_c - \mathcal{F}_r(w_r),\label{eq:r_layer_APIR_outer}
	\end{align}
\end{subequations}
where $U_c \in \mathbb{R}$ is the commanded input to the outermost APIR layer. 

Setting $r=0$ recovers the single-layer \eqref{eq:single_APIR}, and setting $r=1$ with $w_0=u$, $w_1=v$, $(p_{1,0},p_{2,0},\gamma_0) = (p_1,p_2,\gamma)$, and $(p_{1,1},p_{2,1},\gamma_1) = (q_1,q_2,\mu)$ recovers the two-layer cascaded APIR \eqref{eq:cascaded_APIR}.

	Note that each layer in \eqref{eq:r_layer_APIR} is structurally identical to the single-layer APIR \eqref{eq:single_APIR}, with the $k$-th layer state $w_k$ acting as the commanded input to Layer $k-1$. The continuity property established in \Cref{rem:continuity}, the smoothness characterization of \Cref{prop:smoothness_order}, and the parametric sensitivity analysis of \Cref{subsec:param_sensitivity} therefore apply to each layer independently, with the parameter triplet $(p_{1,k}, p_{2,k}, \gamma_k)$ governing layer $k$ in the same interpretable manner as $(p_1,p_2,\gamma)$ governs the single-layer APIR.

The following assumptions characterize the admissibility and smoothness requirements of the generalized cascaded APIR.
\begin{assumption}[Initial admissibility]
	\label{assum:extended_IC}
	The initial values of all auxiliary APIR states satisfy
	\begin{equation}\label{eq:extended_IC}
		w_k(0) \in \mathbb{W}_k\coloneqq (w_{k,\min}, w_{k,\max}), \quad k = 0, 1, \ldots, r.
	\end{equation}
\end{assumption}
\begin{remark}
	The auxiliary states $w_1, \ldots, w_r$ are internal variables of the cascaded APIR and are not physical plant states. Since $0\in\mathbb W_k$ for every layer, they may always be initialized at the origin. Therefore, \Cref{assum:extended_IC} imposes no practical restriction on the initial state of the plant.
\end{remark}
\begin{assumption}[Layer smoothness]
	\label{assum:layer_smoothness}
	The smoothness exponent of each layer satisfies
	\begin{equation}\label{eq:smoothness_req}
		\gamma_k \geq r - k, \, k = 0, 1, \ldots, r-1,
	\end{equation}
	with each $\gamma_k$ an even integer.
\end{assumption}
\begin{remark}
	\Cref{prop:smoothness_order} guarantees that each realization layer possesses sufficient differentiability for the recursive analysis of the cascaded APIR. Since every $\gamma_k$ is a user-selected design parameter, the required smoothness can always be achieved by suitably choosing $\gamma_k$. Therefore, this assumption is compatible with the standard smoothness requirement imposed by recursive backstepping designs and introduces no additional restriction in practice.
\end{remark}
\begin{assumption}[Recursive derivative feasibility]
	\label{assum:consecutive_feasibility}
	For each $k = 1, \ldots, r$, the prescribed output-rate limits of layer $k-1$, denoted $\dot{w}_{k-1,\max} > 0$ and $\lvert\dot{w}_{k-1,\min}\rvert > 0$, satisfy 
	\begin{align}
		\dot{w}_{k-1,\max} &> p_{1,k-1}p_{2,k-1}\lvert w_{k-1,\min}\rvert, \nonumber\\
		\lvert\dot{w}_{k-1,\min}\rvert &> p_{1,k-1}p_{2,k-1}w_{k-1,\max}. \label{eq:consec_feasibility}
	\end{align}
\end{assumption}
The admissible bounds of Layer $k$ are then chosen under \Cref{assum:consecutive_feasibility} as
\begin{align}
	w_{k,\max}&\coloneqq\frac{\dot{w}_{k-1,\max}}{p_{1,k-1}} - p_{2,k-1}\lvert w_{k-1,\min}\rvert,\nonumber\\
	w_{k,\min}&\coloneqq-\frac{\lvert\dot{w}_{k-1,\min}\rvert}{p_{1,k-1}}+ p_{2,k-1}w_{k-1,\max}, \label{eq:general_layer_bounds}
\end{align}
which ensures $w_{k,\max} > 0$ and $w_{k,\min} < 0$, ensuring that each APIR layer in the cascade is well-posed.
\begin{remark}
	\Cref{assum:consecutive_feasibility} recursively exposes the compatibility between each layer's admissible range and the requested output-rate bounds. Reducing $p_{2,k-1}$ can enlarge feasibility, but it also changes the interior margin and gain conditioning of Layer $k-1$. The conditions should therefore be treated as explicit tuning tradeoffs.
\end{remark}
The following theorem establishes the fundamental properties of the proposed cascaded APIR. 
\begin{theorem}[Admissibility of the r-layer APIR]
	\label{thm:r_layer_cascade}
	Consider the $r$-layer cascaded APIR \eqref{eq:r_layer_APIR} under \Cref{assum:extended_IC,assum:layer_smoothness,assum:consecutive_feasibility}. Suppose the commanded input satisfies $\lvert U_c(t) \rvert \leq \Xi < \infty$ for all $t \geq 0$. Then, the following holds for all $t \geq 0$
	\begin{itemize}
		\item \emph{(Layer admissibility)} The admissible set of each APIR layer is forward invariant, that is, $w_k(t) \in \mathbb{W}_k$ for every $k = 0, 1, \ldots, r$.
		\item \emph{(Layer-output rate admissibility)} The time derivative of each intermediate APIR state is forward invariant within its prescribed rate set: $\dot{w}_{k-1}(t) \in \mathbb{S}_{k-1}$ for every $k = 1, \ldots, r$, where $\mathbb{S}_{k-1} =\left(\dot{w}_{k-1,\min}, \dot{w}_{k-1,\max}\right)$.
	\end{itemize}
\end{theorem}
\begin{proof}
	The proof is divided into two parts.
	
	\textit{Part (i): Layer admissibility.} We prove this by backward induction on layer index $k$.
	
	\textit{Base case ($k = r$, outermost APIR layer):} Layer $r$ is governed by \eqref{eq:r_layer_APIR_outer}, which is structurally identical to the single-layer APIR \eqref{eq:single_APIR} under the substitutions $u \mapsto w_r$, $u_c \mapsto U_c$, $p_1 \mapsto p_{1,r}$, $p_2 \mapsto p_{2,r}$, $\gamma \mapsto \gamma_r$, $u_{\min} \mapsto w_{r,\min}$, and $u_{\max} \mapsto w_{r,\max}$. Since $|U_c(t)| \leq \Xi < \infty$ for all $t \geq 0$ and $w_r(0) \in \mathbb{W}_r$ by \Cref{assum:extended_IC} and \Cref{thm:APIR_admissibility} applied to \eqref{eq:r_layer_APIR_outer} establishes $w_r(t) \in \mathbb{W}_r$ for all $t \geq 0$.
	
	\textit{Inductive step:}
	Suppose $w_m(t) \in \mathbb{W}_m$ for some $1 \leq m \leq r$ and all $t \geq 0$. Since $\mathbb{W}_m = (w_{m,\min}, w_{m,\max})$ is a bounded open interval, there exists a finite constant
	\begin{equation}\label{eq:Wm_star}
		W_m^* \coloneqq \max\left(\lvert w_{m,\min}\rvert, w_{m,\max}\right) < \infty,
	\end{equation}
	such that $\lvert w_m(t)\rvert \leq W_m^*$ for all $t \geq 0$. Layer $m-1$ is governed by \eqref{eq:r_layer_APIR_inner} with $k = m$, which has the same structure as the single-layer APIR \eqref{eq:single_APIR}. Applying \Cref{thm:APIR_admissibility} to layer $m-1$ with the bounded commanded input $|w_m(t)| \leq W_m^*$ and initial condition $w_{m-1}(0) \in \mathbb{W}_{m-1}$ (from \Cref{assum:extended_IC}) yields $w_{m-1}(t) \in \mathbb{W}_{m-1}$ for all $t \geq 0$. Applying this argument successively for $m = r, r-1, \ldots, 1$ establishes $w_k(t) \in \mathbb{W}_k$ for all $k = 0, 1, \ldots, r$ and all $t \geq 0$. Thus, the admissible set of every APIR layer in the cascade is forward invariant.
	
	\textit{Part (ii): Layer-output rate admissibility.}
	Set $k \in \{1, \ldots, r\}$. By part (i), $w_{k-1}(t) \in \mathbb{W}_{k-1}$ and $w_k(t) \in \mathbb{W}_k$ for all $t \geq 0$. From  \eqref{eq:GkFk_def} and \eqref{eq:r_layer_APIR_inner},
	\begin{align}
		\dot{w}_{k-1}
		&= \mathcal{G}_{k-1}(w_{k-1})w_k - \mathcal{F}_{k-1}(w_{k-1}), \nonumber\\
		&= p_{1,k-1}\left[\mathcal{S}_{k-1}(w_{k-1})w_k - p_{2,k-1}w_{k-1}\right]. \label{eq:layer_k_rate_dyn}
	\end{align}
	Note that the function satisfies $\mathcal{S}_{k-1}(w_{k-1}) \in (0, 1)$ for all $w_{k-1} \in \mathbb{W}_{k-1}$. We establish upper and lower bounds on $\dot w_{k-1}$.
	
	\textit{Upper bound of $\dot{w}_{k-1}$:}
	Since $\mathcal{S}_{k-1}(w_{k-1}) < 1$ and $w_k(t) < w_{k,\max}$, the term $\mathcal{S}_{k-1}(w_{k-1})w_k < w_{k,\max}$. As $w_{k-1}(t) > w_{k-1,\min}$ and $p_{2,k-1} > 0$, we have $-p_{2,k-1}w_{k-1} < p_{2,k-1}|w_{k-1,\min}|$. Substituting these into \eqref{eq:layer_k_rate_dyn}, one may obtain the upper bound as
	\begin{equation}\label{eq:rate_upper}
		\dot{w}_{k-1} < p_{1,k-1}\left(w_{k,\max} + p_{2,k-1}\lvert w_{k-1,\min}\rvert \right) = \dot{w}_{k-1,\max}.
	\end{equation}
	
	\textit{Lower bound of $\dot{w}_{k-1}$:}
	If $w_k \geq 0$, then $\mathcal{S}_{k-1}(w_{k-1})w_k \geq 0 > w_{k,\min}$. If $w_k < 0$, then since $\mathcal{S}_{k-1}(w_{k-1}) \in (0,1)$ and $w_k > w_{k,\min}$, we have $\mathcal{S}_{k-1}(w_{k-1})w_k > w_{k,\min}$ (the inequality reverses because $w_{k,\min} < 0$ and $\mathcal{S}_{k-1} < 1$). In both cases, $\mathcal{S}_{k-1}(w_{k-1})w_k > w_{k,\min}$. Since $w_{k-1}(t) < w_{k-1,\max}$, we have $-p_{2,k-1}w_{k-1} > -p_{2,k-1}w_{k-1,\max}$. Substituting into \eqref{eq:layer_k_rate_dyn},
	\begin{equation}\label{eq:rate_lower}
		\dot{w}_{k-1}> p_{1,k-1}\left(w_{k,\min} - p_{2,k-1}w_{k-1,\max}\right)= \dot{w}_{k-1,\min}.
	\end{equation}
	Combining \eqref{eq:rate_upper} and \eqref{eq:rate_lower} establishes $\dot{w}_{k-1}(t) \in \mathbb{S}_{k-1}$ for every $k=1,\ldots,r$. This completes the proof.
\end{proof}

\section{Numerical Simulations}\label{sec:simulations}
We evaluate the proposed APC framework from four complementary perspectives: tracking under asymmetric actuator-magnitude limits, simultaneous input--output admissibility, enforcement of actuator magnitude and rate constraints, and numerical conditioning of the APIR-based command maps. The purpose of our study is not merely to show convergence of the tracking error, but to examine how the closed-loop trajectories use the available asymmetric control authority while retaining a nonzero distance from the prescribed safety boundaries. In the simulations, we evaluate both constraint satisfaction and the nonsingularity mechanism underlying the compatibility conditions in \Cref{thm:global_tracking,thm:semi_global_tracking}.

We consider
\begin{align*}
	\dot{x}_{1}=&~0.1x_{1}^{2}+x_{2},\\
	\dot{x}_{2}=&~0.1x_{1}x_{2}-0.2x_{1}+(1+x_{1}^{2})u,
\end{align*}
with $u_{\min}=-0.5$, $u_{\max}=0.75$, and
$y_d(t)=0.2+0.3\sin(t)$. Unless stated otherwise,
$p_1=100$, $p_2=0.1$, $\gamma=2$, and
$k_1=k_2=k_3=2$. The APIR state is initialized at $u(0)=0$.
For every case, we record the minimum actuator margin, the minimum realization gain, the peak commanded input, and the final tracking error. These trajectory-level diagnostics do not replace the set-based compatibility certificate in \eqref{eq:self_consistency}, but they verify that the simulated trajectories remain in a nonsingular compact interior region. The initial conditions used in the reported cases are summarized in \Cref{tab:simulation_cases}.
\begin{table}[!t]
	\centering
	\caption{Simulation cases and initial conditions.}
	\label{tab:simulation_cases}
	\resizebox{\linewidth}{!}{%
		\begin{tabular}{lll}
			\toprule
			Scenario & Case labels & Initial plant states $\left(x_{1}(0),x_{2}(0)\right)$ \\
			\midrule
			Asymmetric magnitude constraint
			& $C_{1}$, $C_{2}$, $C_{3}$
			& $(0,0)$, $(-0.2,0.25)$, $(0.6,-0.1)$ \\
			Time-varying output-safe tracking
			& $I_{1}$, $I_{2}$, $I_{3}$
			& $(0.1,0.538)$, $(0.25,0.174)$, $(0.4,-0.156)$ \\
			Constant output-safe tracking
			& $S_{1}$, $S_{2}$
			& $(0.1,0.509)$, $(0.35,-0.03)$ \\
			Magnitude--rate constraints
			& $R_{1}$, $R_{2}$, $R_{3}$
			& $(0.05,0.6)$, $(0.2,0.296)$, $(0.35,-0.012)$ \\
			\bottomrule
	\end{tabular}}
\end{table}
For interpretation, we define the trajectory-wise safety margins
\begin{align*}
	m_{u}(t) =&~ \min\left\{u(t)-u_{\min},\,u_{\max}-u(t)\right\},\\
	m_{y}(t) =&~ \min\left\{y(t)-\underline{y}(t),\,
	\overline{y}(t)-y(t)\right\},\\
	m_{\dot{u}}(t) =&~ \min\left\{\dot{u}(t)-\dot{u}_{\min},\,
	\dot{u}_{\max}-\dot{u}(t)\right\}.
\end{align*}
Positive values of these quantities certify strict interiority along a simulated trajectory. We also monitor $\mathcal{G}(u)$ and, for the cascaded case, $\mathcal{G}_{v}(v)$. Positive lower bounds on these gains quantify the conditioning of the corresponding commanded-input maps.

We first demonstrate the role of the APIR in output tracking under asymmetric actuator-magnitude limits. The trajectories show two distinct closed-loop signals with different physical roles. The backstepping law produces the commanded input $u_{c}$ required by the recursive stabilization argument, whereas the APIR generates the realized plant input $u$. When the command demand is moderate, the APIR behaves nearly linearly. As the demand increases, its effective gain decreases, and the realized input approaches the relevant actuator boundary smoothly. Consequently, the available positive and negative authorities are used according to their actual unequal limits, rather than being reduced to a conservative symmetric interval.

The most demanding trajectory approaches the upper actuator boundary while retaining a strictly positive input margin. At the same time, the minimum sampled value of $\mathcal{G}(u)$ remains positive, so the commanded-input map does not encounter the singularity associated with a vanishing APIR gain. The tracking errors for all three cases decay to zero after the transient. One may observe that, in addition to guaranteeing $u$ is bounded, the APIR converts a potentially larger command demand into a continuously differentiable, strictly interior plant input while preserving sufficient control authority for tracking.
\begin{figure*}[!ht]
	\centering
	\begin{subfigure}[t]{0.245\linewidth}
		\centering
		\includegraphics[width=\linewidth]{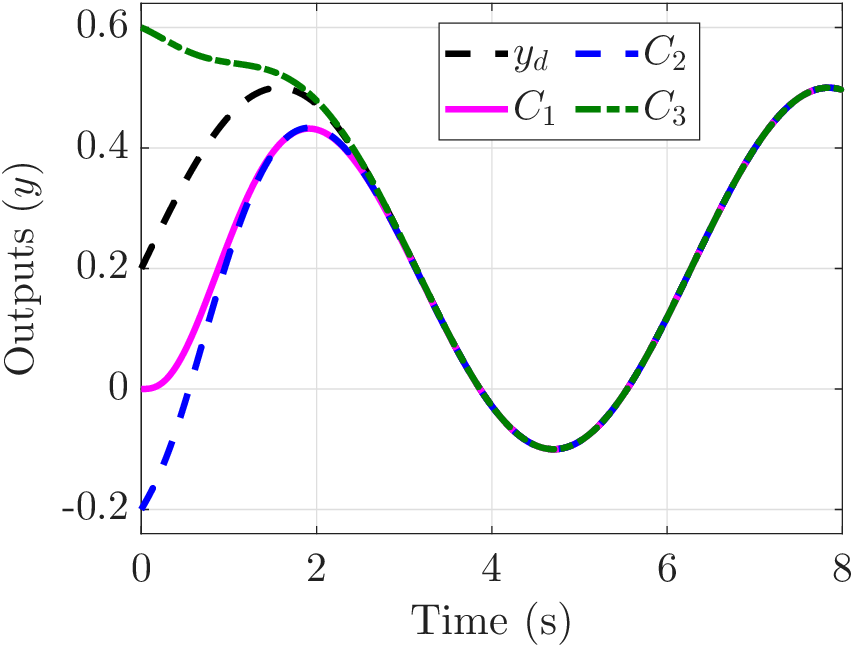}
		\caption{Output trajectories.}
		\label{fig:Bounded_Input_output}
	\end{subfigure}%
	\begin{subfigure}[t]{0.245\linewidth}
		\centering
		\includegraphics[width=\linewidth]{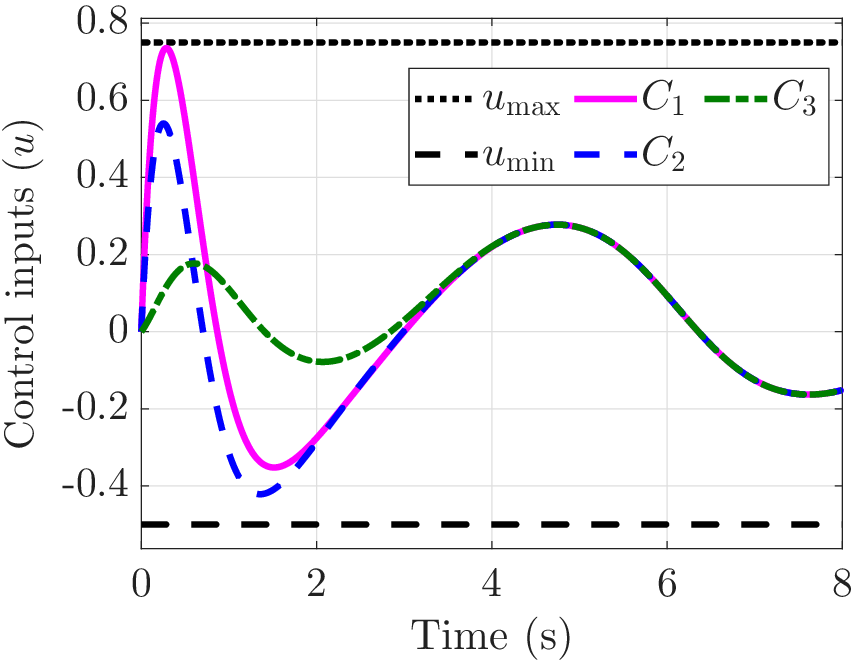}
		\caption{Control inputs.}
		\label{fig:Bounded_Input_input}
	\end{subfigure}
	\begin{subfigure}[t]{0.245\linewidth}
		\centering
		\includegraphics[width=\linewidth]{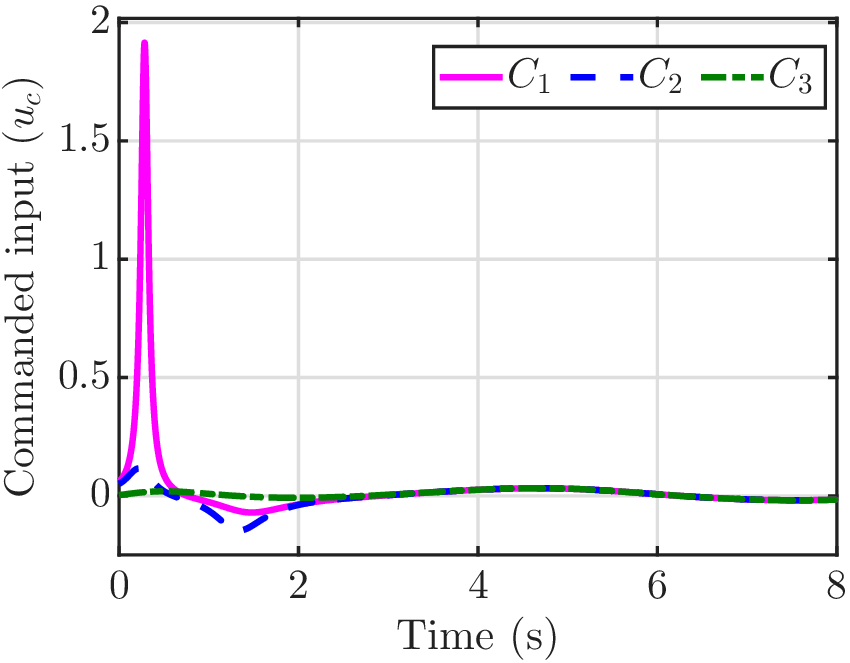}
		\caption{Commanded inputs.}
		\label{fig:Bound_Input_commanded_input}
	\end{subfigure}%
	\begin{subfigure}[t]{0.245\linewidth}
		\centering
		\includegraphics[width=\linewidth]{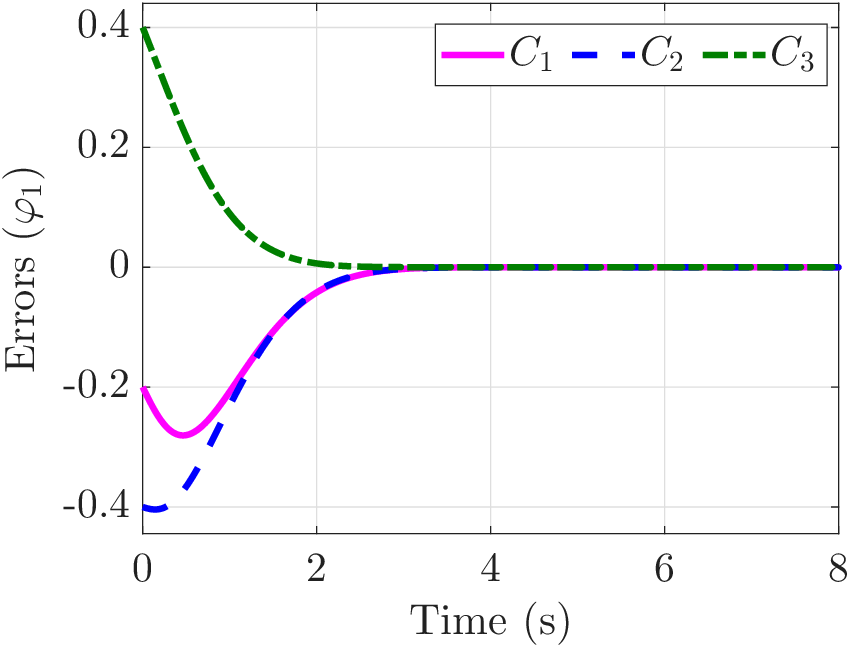}
		\caption{Tracking errors.}
		\label{fig:Bounded_Input_error}
	\end{subfigure}
	\caption{Output tracking with bounded input.}
	\label{fig:Bounded_Input}
\end{figure*}
\begin{figure*}[!ht]
	\centering
	\begin{subfigure}[t]{0.245\linewidth}
		\centering
		\includegraphics[width=\linewidth]{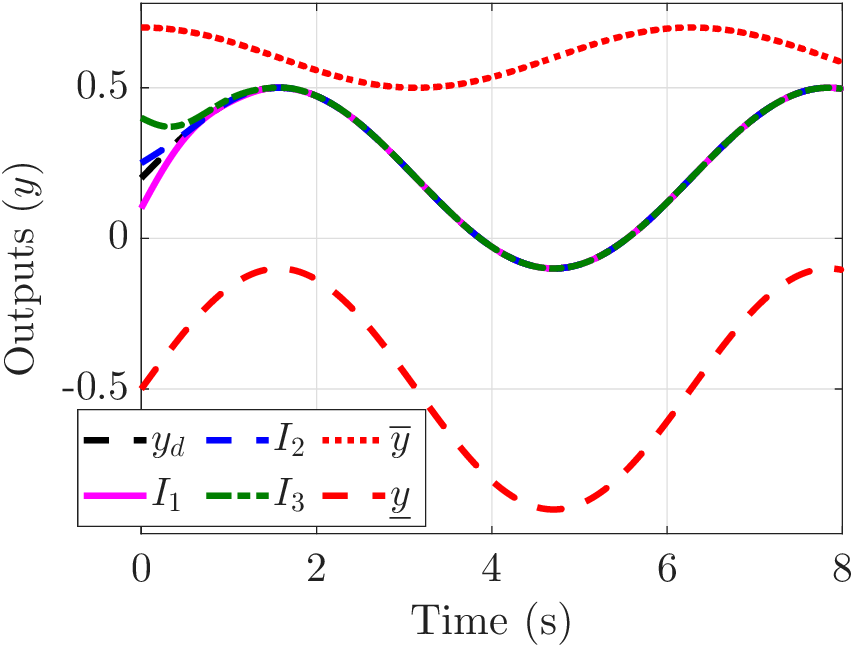}
		\caption{Output trajectories.}
		\label{fig:BI_Time_var_Output_output}
	\end{subfigure}%
	\begin{subfigure}[t]{0.245\linewidth}
		\centering
		\includegraphics[width=\linewidth]{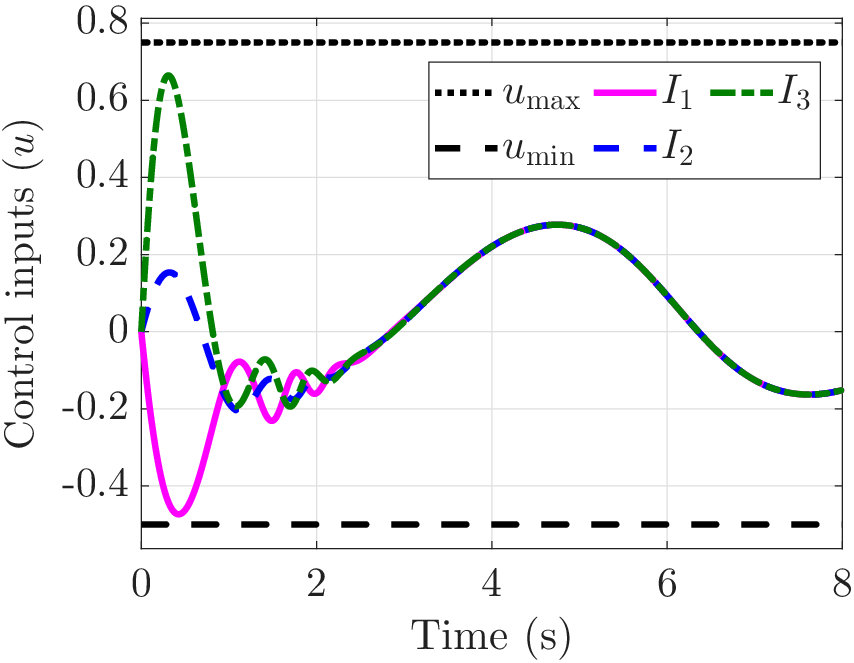}
		\caption{Control inputs.}
		\label{fig:BI_Time_var_Output_input}
	\end{subfigure}
	\begin{subfigure}[t]{0.245\linewidth}
		\centering
		\includegraphics[width=\linewidth]{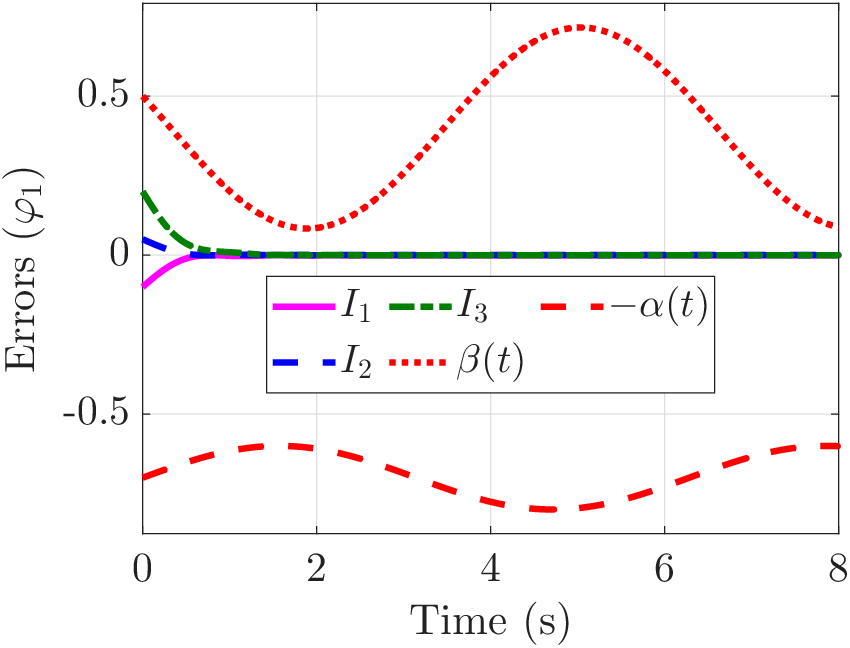}
		\caption{Tracking errors.}
		\label{fig:BI_Time_var_Output_error}
	\end{subfigure}%
	\begin{subfigure}[t]{0.245\linewidth}
		\centering
		\includegraphics[width=\linewidth]{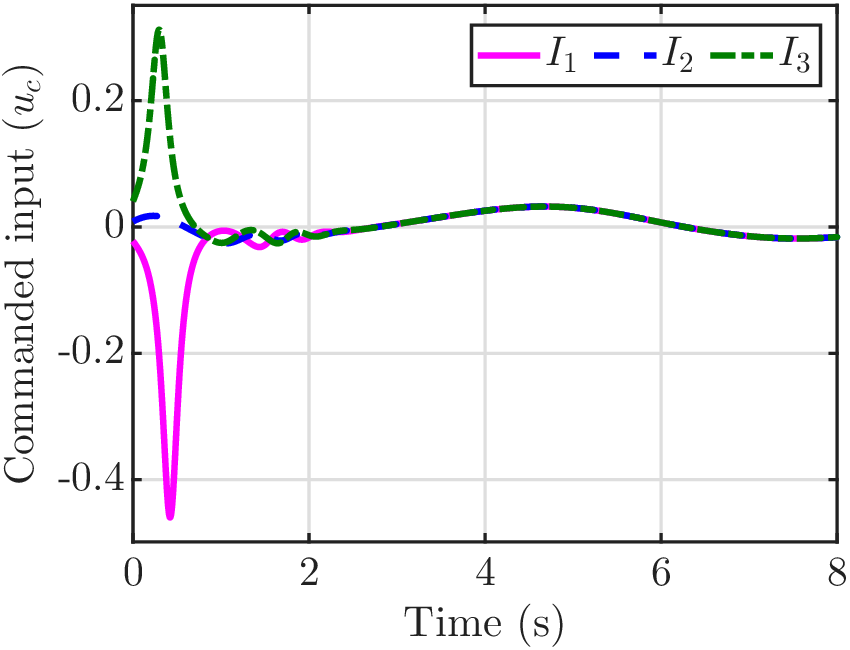}
		\caption{Commanded inputs.}
		\label{fig:BI_Time_var_Output_commanded_input}
	\end{subfigure}
	\caption{Constrained (time-varying) output tracking with bounded input.}
	\label{fig:BI_Time_var_Output}
\end{figure*}

We next examine whether actuator admissibility can be composed with a time-varying output-safe set without altering the APIR mechanism. The output boundaries are selected as $\overline{y}(t)=0.6+0.1\cos(t)$ and $\underline{y}(t)=-0.5+0.4\sin(t)$. The smooth logarithmic barrier coordinate maps the moving admissible interval to the unconstrained coordinate $z_{1}$. This provides a useful numerical interpretation of the design. As the output approaches either boundary, a fixed change in the physical tracking error produces an increasingly large change in $z_{1}$, and the recursive controller reacts before boundary contact occurs. Unlike a branch-selected barrier construction, the coordinate remains smooth when the tracking error changes sign.

The trajectories remain inside the moving output corridor while the APIR simultaneously preserves the asymmetric actuator interval. The smallest output margin occurs near the upper boundary, whereas the lower boundary remains comparatively inactive. This asymmetry is informative. The controller does not artificially center the trajectory inside the safe corridor, but uses the available output region according to the reference demand. Moreover, the APIR gain remains separated from zero throughout the same transient, showing that output-safety enforcement does not achieve its objective by forcing the actuator realization into a poorly conditioned boundary layer.

To separate the effect of boundary motion from the barrier mechanism itself, we next fix the output interval at $\underline{y}=-0.5$ and $\overline{y}=0.7$. No redesign is required. The same barrier-coordinate construction applies after setting the boundary derivatives to zero. Relative to the time-varying case, the constant corridor produces a larger minimum distance from both the output and actuator boundaries and a smaller command demand. This comparison illustrates that the additional control effort in the moving-boundary case is associated with tracking the geometry of the admissible set, rather than with a change in the underlying APC architecture.

\begin{figure*}[!ht]
	\centering
	\begin{subfigure}[t]{0.245\linewidth}
		\centering
		\includegraphics[width=\linewidth]{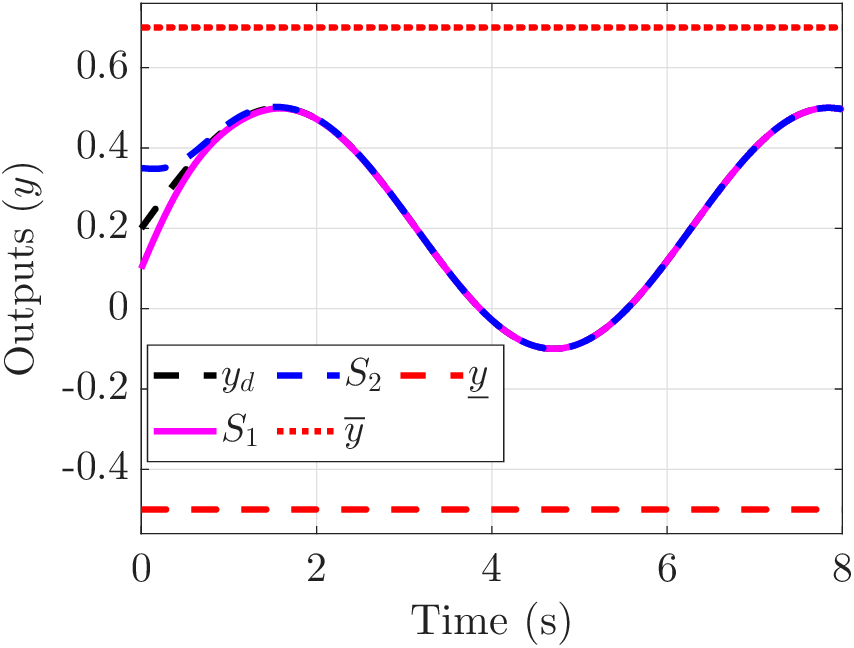}
		\caption{Output trajectories.}
		\label{fig:BI_Constant_Output_output}
	\end{subfigure}%
	\begin{subfigure}[t]{0.245\linewidth}
		\centering
		\includegraphics[width=\linewidth]{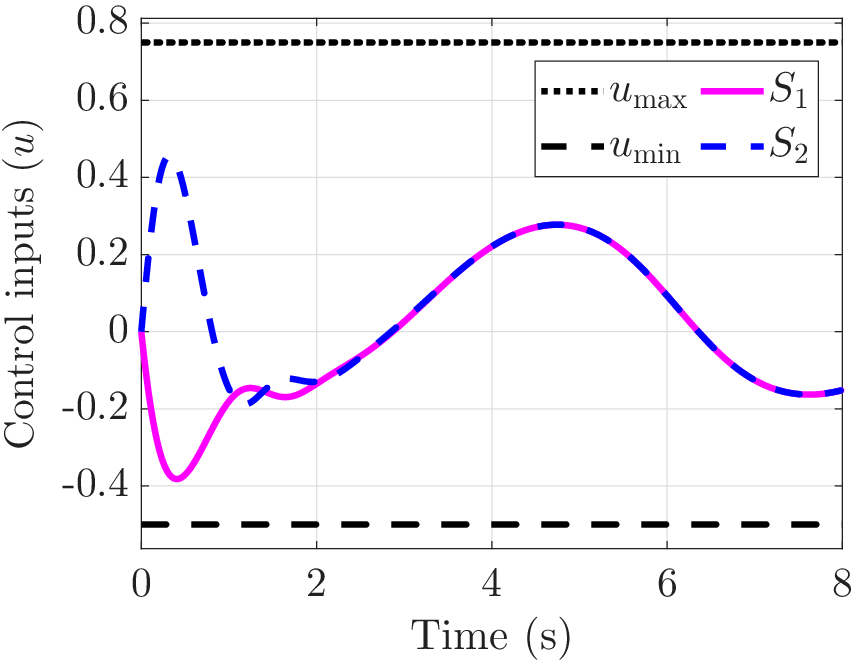}
		\caption{Control inputs.}
		\label{fig:BI_Constant_Output_input}
	\end{subfigure}
	\begin{subfigure}[t]{0.245\linewidth}
		\centering
		\includegraphics[width=\linewidth]{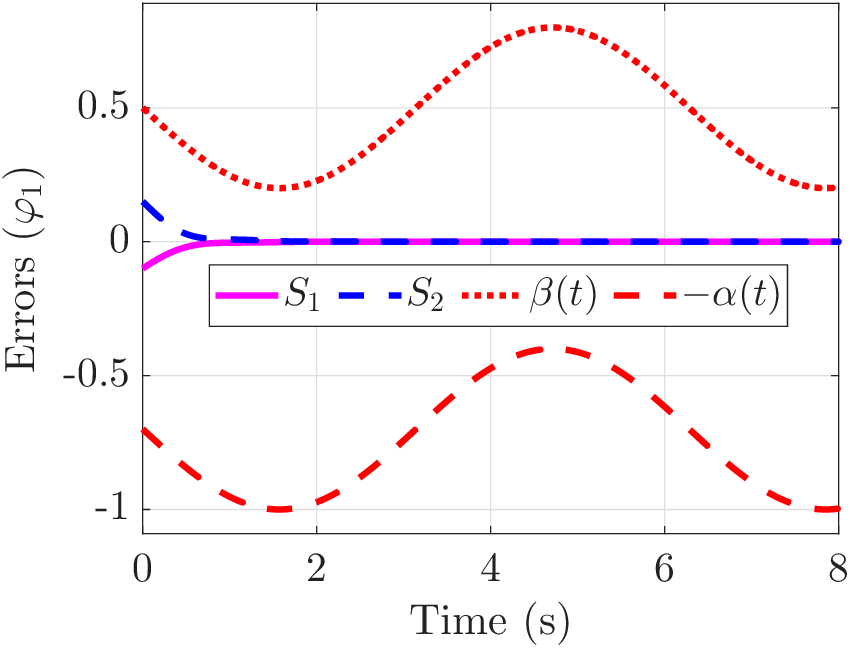}
		\caption{Tracking errors.}
		\label{fig:BI_Constant_Output_error}
	\end{subfigure}%
	\begin{subfigure}[t]{0.245\linewidth}
		\centering
		\includegraphics[width=\linewidth]{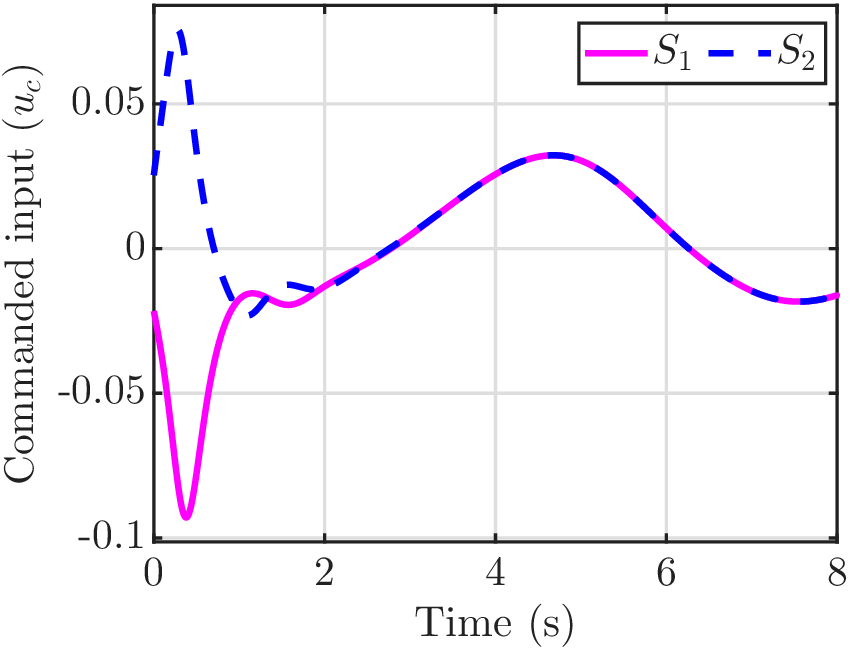}
		\caption{Commanded inputs.}
		\label{fig:BI_Constant_Output_commanded_input}
	\end{subfigure}
	\caption{Constrained (constant) output tracking with bounded input.}
	\label{fig:BI_Constant_Output}
\end{figure*}
\begin{figure*}[!ht]
	\centering
	\begin{subfigure}[t]{0.245\linewidth}
		\centering
		\includegraphics[width=\linewidth]{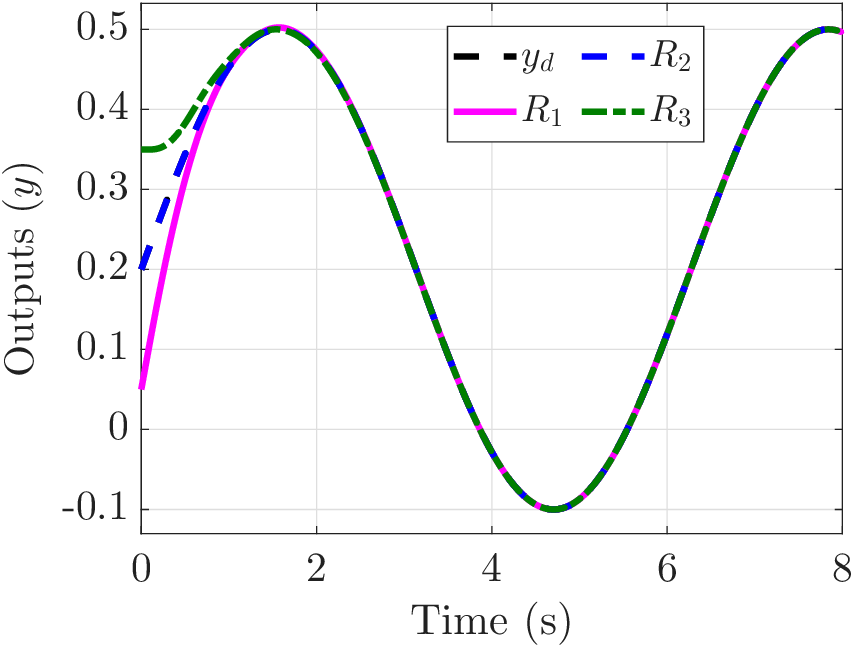}
		\caption{Output trajectories.}
		\label{fig:BI_rate_output}
	\end{subfigure}%
	\begin{subfigure}[t]{0.245\linewidth}
		\centering
		\includegraphics[width=\linewidth]{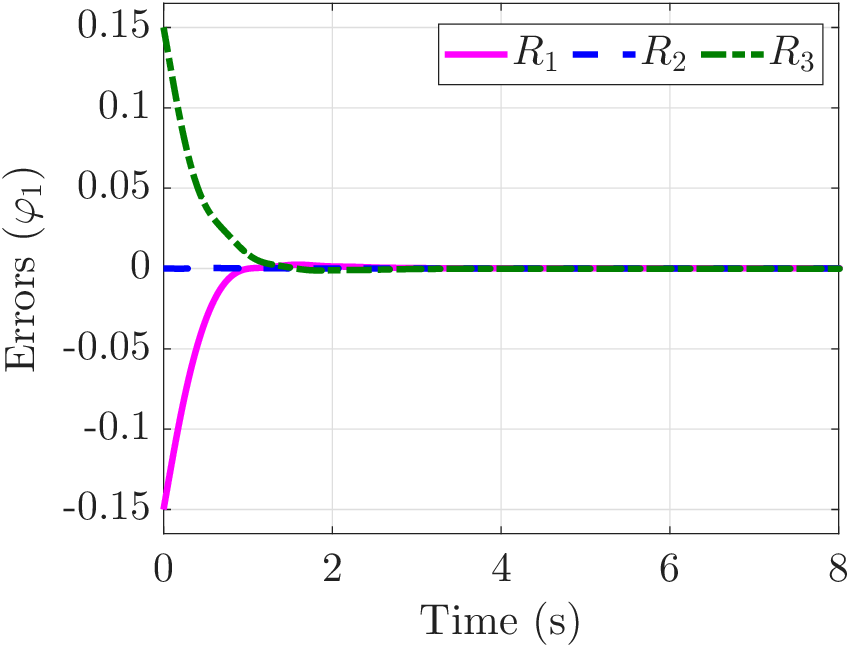}
		\caption{Tracking errors.}
		\label{fig:BI_rate_error}
	\end{subfigure}
	\begin{subfigure}[t]{0.245\linewidth}
		\centering
		\includegraphics[width=\linewidth]{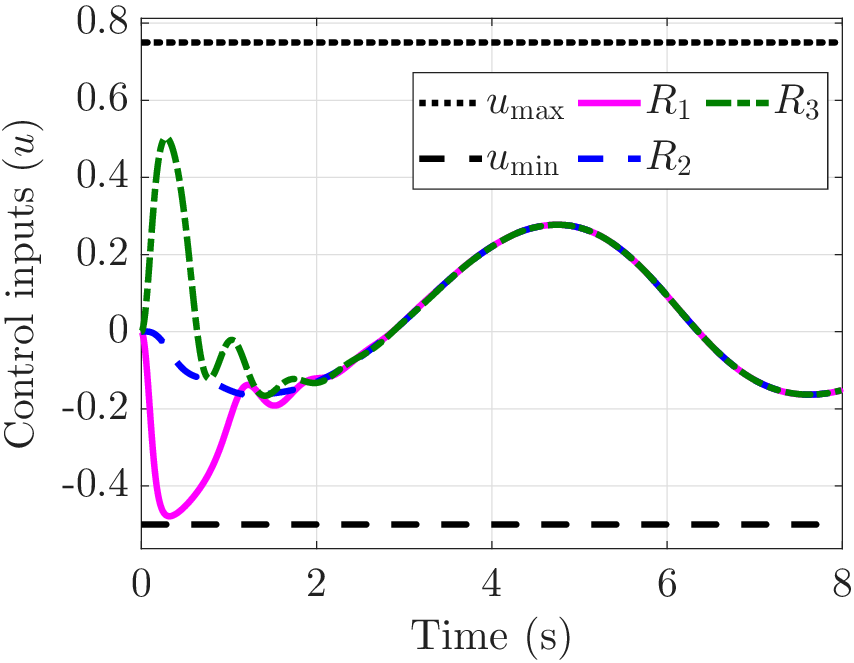}
		\caption{Control inputs.}
		\label{fig:BI_rate_input}
	\end{subfigure}%
	\begin{subfigure}[t]{0.245\linewidth}
		\centering
		\includegraphics[width=\linewidth]{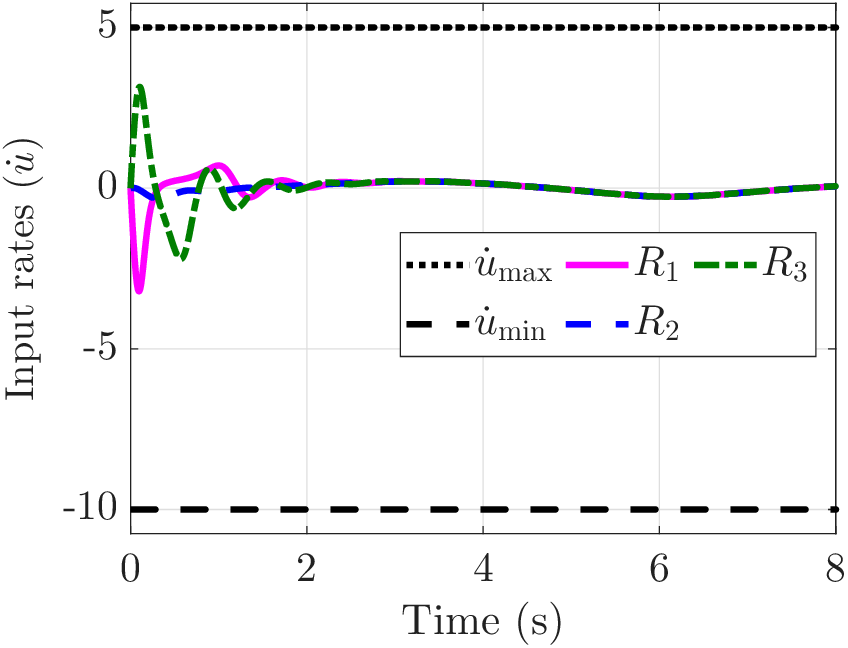}
		\caption{Control input rates.}
		\label{fig:BI_rate_rate}
	\end{subfigure}
	\caption{Output tracking with bounded input and rate.}
	\label{fig:BI_Rate}
\end{figure*}
\begin{figure*}[!ht]
	\centering
	\begin{subfigure}[t]{0.245\linewidth}
		\centering
		\includegraphics[width=\linewidth]{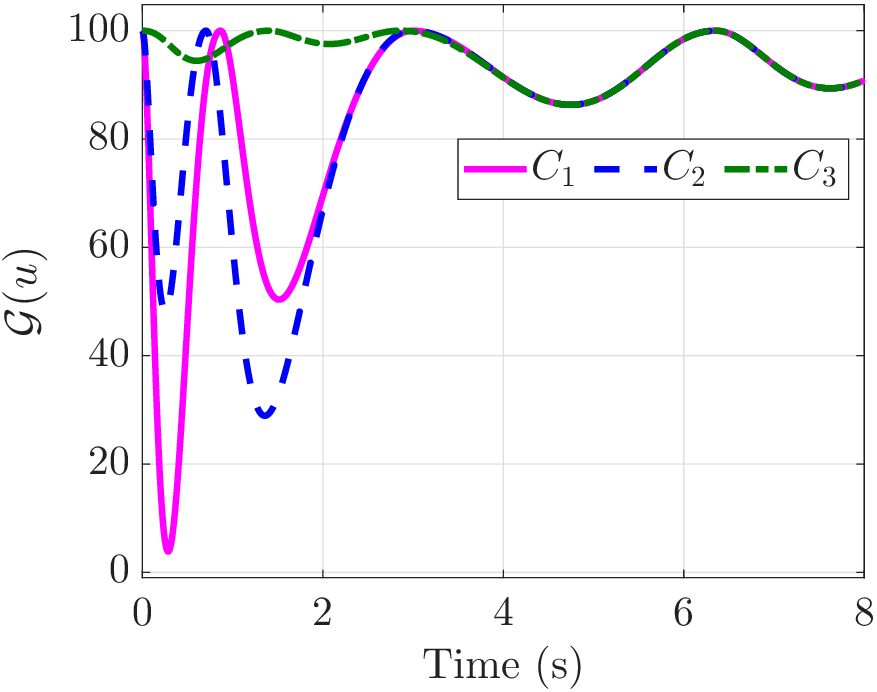}
		\caption{Magnitude-only APIR gain.}
		\label{fig:APC_gain_conditioning_sub1}
	\end{subfigure}%
	\begin{subfigure}[t]{0.245\linewidth}
		\centering
		\includegraphics[width=\linewidth]{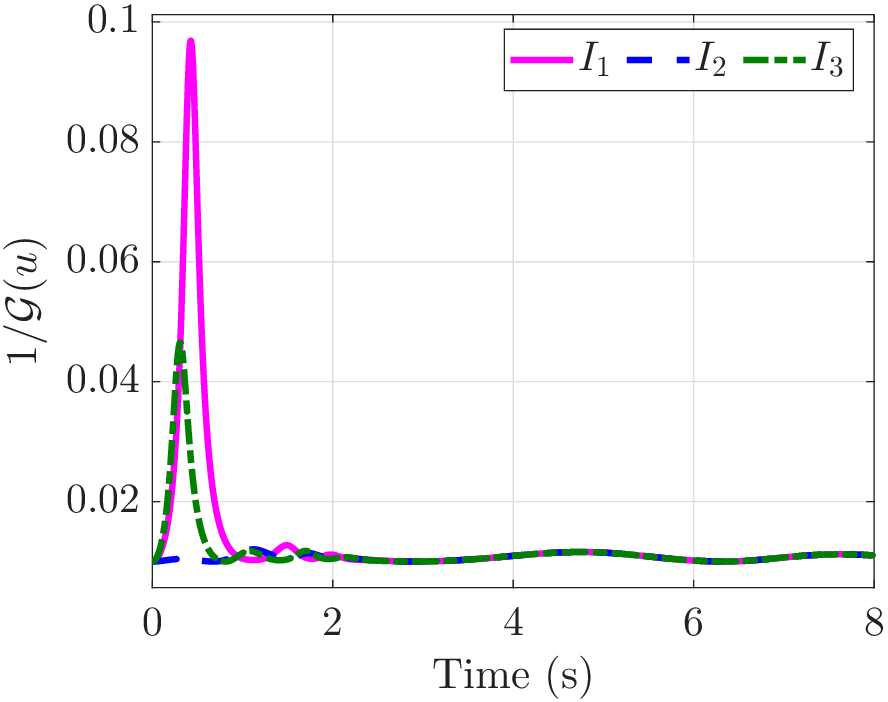}
		\caption{Output-safe command conditioning.}
		\label{fig:APC_gain_conditioning_sub2}
	\end{subfigure}
	\begin{subfigure}[t]{0.245\linewidth}
		\centering
		\includegraphics[width=\linewidth]{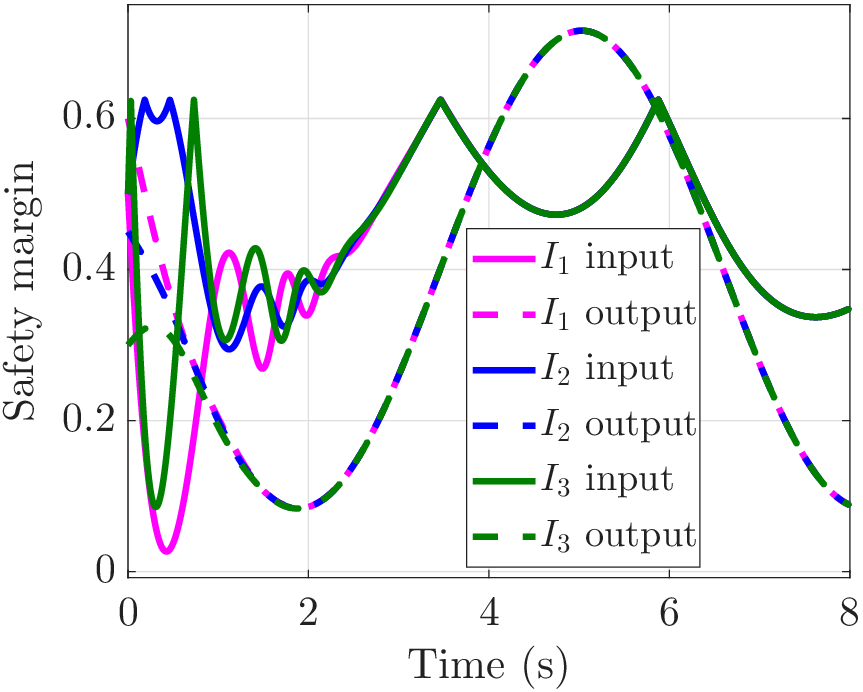}
		\caption{Trajectory-wise safety margins.}
		\label{fig:APC_gain_conditioning_sub3}
	\end{subfigure}%
	\begin{subfigure}[t]{0.245\linewidth}
		\centering
		\includegraphics[width=\linewidth]{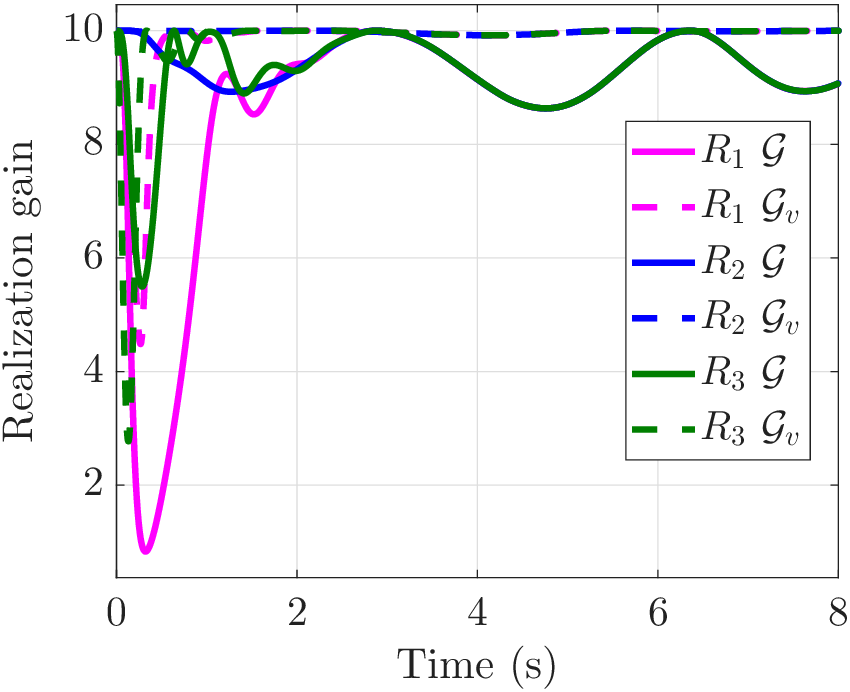}
		\caption{Cascaded-APIR gains.}
		\label{fig:APC_gain_conditioning_sub4}
	\end{subfigure}
	\caption{APC gain conditioning.}
	\label{fig:APC_gain_conditioning}
\end{figure*}
\begin{figure}[!ht]
	\centering
	\begin{subfigure}[t]{0.48\linewidth}
		\centering
		\includegraphics[width=\linewidth]{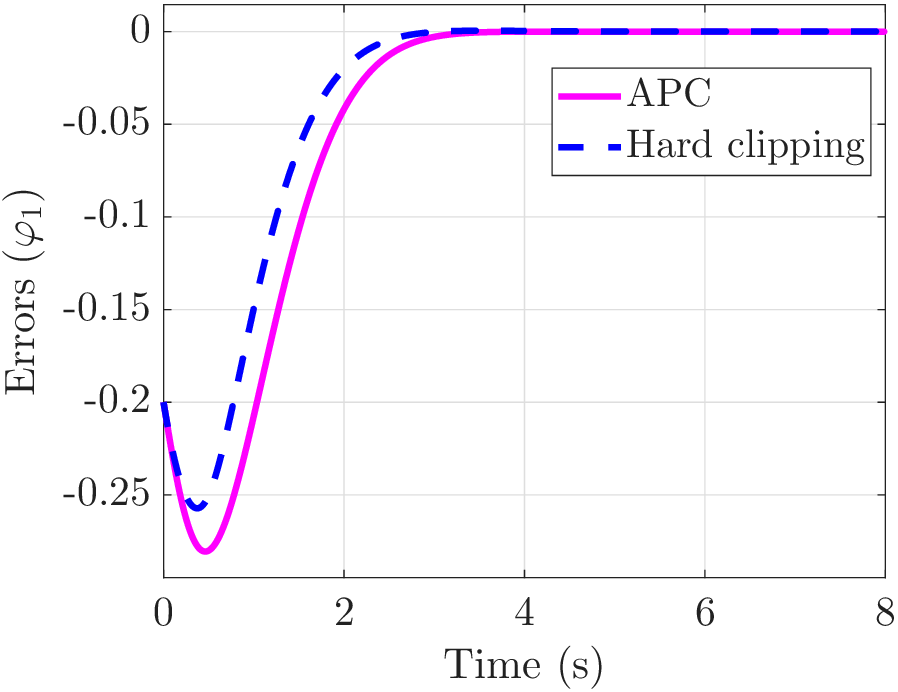}
		\caption{Tracking error comparison.}
		\label{fig:APC_baseline_comparison_sub1}
	\end{subfigure}%
	\begin{subfigure}[t]{0.48\linewidth}
		\centering
		\includegraphics[width=\linewidth]{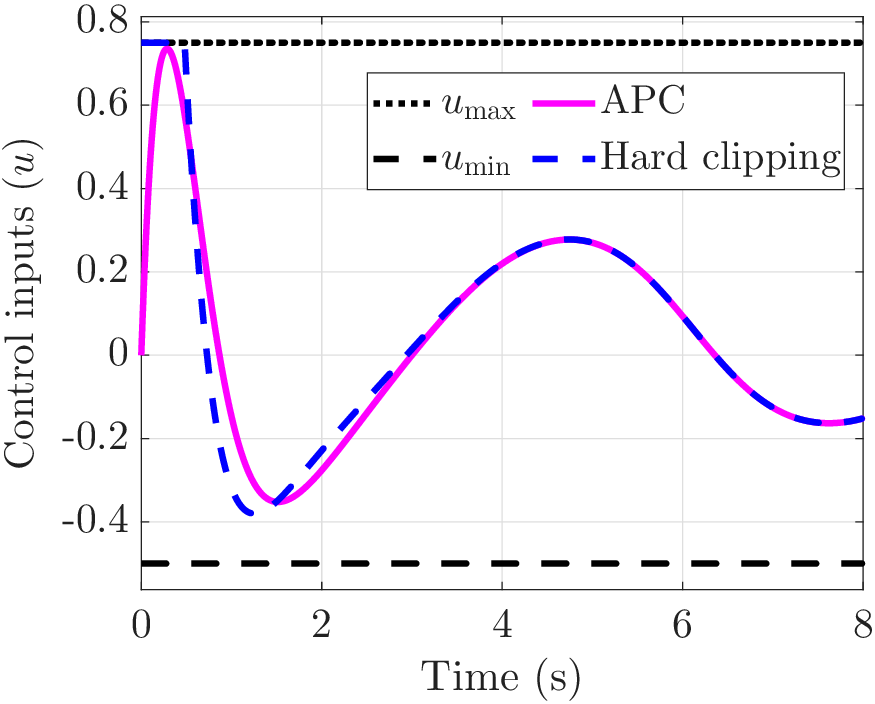}
		\caption{Control inputs.}
		\label{fig:APC_baseline_comparison_sub2}
	\end{subfigure}
	\caption{APC baseline comparison.}
	\label{fig:APC_baseline_comparison}
\end{figure}

We further evaluate the two-layer APIR architecture for cascaded enforcement of magnitude and rate limits. The magnitude-layer
parameters are selected as $p_{1}=10$, $p_{2}=0.1$, and $\gamma=2$, while the rate-layer parameters are $q_{1}=10$, $q_{2}=0.1$, and $\mu=2$. The physical rate limits are $\dot{u}_{\min}=-10$ and $\dot{u}_{\max}=5$. From the compatibility relation in \eqref{eq:rate_layer_bounds}, the corresponding internal-layer bounds are $v_{\min}=-0.925$ and $v_{\max}=0.45$, with $u(0)=v(0)=0$.

The cascade separates two safety functions. The outer realization confines the auxiliary state $v$, while the inner realization converts $v$ into an input trajectory whose magnitude and rate both remain admissible. The simulation shows that neither layer needs to operate as a hard clip. Instead, the internal state moves within its asymmetric interval and continuously modulates the rate of the physical input. The realized input, therefore, remains inside $\mathbb{U}$ while $\dot{u}$ stays inside its prescribed asymmetric rate set. Both realization gains remain positive, confirming that the extra layer preserves the nonsingularity required for recursive synthesis.
\begin{remark}
	The rate bounds are not merely checked after the magnitude-constrained input has been generated. They are encoded into the allowable range of the intermediate state through \eqref{eq:rate_layer_bounds}. Thus, the observed rate admissibility follows from the architecture of the cascade, not from a posteriori clipping of $\dot{u}$. This is the principal numerical distinction between the cascaded APIR and a conventional magnitude limiter followed by an independent slew-rate limiter.
\end{remark}
\begin{remark}
	Two trends are apparent from our results. First, the minimum APIR gains remain strictly positive even when the realized input approaches a boundary, which supports the continuation argument used in the compatibility analysis. Second, enforcing a moving output corridor or an additional rate constraint reduces the available safety margin, but does not collapse either realization gain. The simulations, therefore, exhibit the intended tradeoff: more restrictive physical requirements consume interior margin, whereas the APIR parameters determine how that margin is distributed between actuator utilization and command conditioning.
\end{remark}

Hard clipping may yield a smaller integral tracking error in some cases by placing the input directly on a constraint boundary. APC instead preserves a strict interior margin and supplies a differentiable actuator state for the recursive design. Hence, the comparison is not intended to establish uniform tracking superiority. It exposes the tradeoff between aggressive boundary utilization and realization-level admissibility. The conditioning plots should be read together with this comparison. A small input margin is acceptable only when the APIR gain remains sufficiently separated from zero, and the commanded-input map remains numerically well conditioned.

\section{Conclusions}\label{sec:conclusions}
The results establish that APC provides a realization-centered means of reconciling nonlinear tracking objectives with finite asymmetric actuator authority. The proposed APIR converts a bounded compatible command into a continuously differentiable plant input that remains strictly inside the prescribed actuator limits, while compact interiority maintains a positive realization gain and prevents singularity of the commanded-input map. When incorporated into recursive backstepping, this structure yields regional asymptotic tracking over APC-compatible sublevels without requiring an ISS assumption on the uncontrolled plant; feasibility is instead captured directly through the compatibility among the desired motion, available actuator authority, and APIR conditioning. The output-constrained extension further establishes that the same realization principle can be composed with a smooth asymmetric logarithmic barrier coordinate, allowing actuator admissibility and time-varying output safety to be enforced within a common closed-loop construction. Likewise, the cascaded APIR shows that admissibility can be propagated across magnitude and rate layers, with the realization states and layer-output rates remaining within their prescribed sets under the corresponding compatibility conditions. The results support APC as a realization-centered constrained-control framework in which APIR realizes admissibility, APC compatibility characterizes usable authority, and the outer synthesis governs closed-loop performance. The strict-feedback construction developed here is one instantiation of this broader realization-centered architecture.

\bibliographystyle{ieeetr}
\bibliography{ref_input_sat.bib}

\begin{thebibliography}{10}

\bibitem{203432}
E.~Sontag and H.~Sussmann, ``Nonlinear output feedback design for linear
  systems with saturating controls,'' in {\em 29th IEEE Conference on Decision
  and Control}, pp.~3414--3416 vol.6, 1990.

\bibitem{doi:10.1080/00207176908905846}
A.~T. Fuller, ``In-the-large stability of relay and saturating control systems
  with linear controllers,'' {\em International Journal of Control}, vol.~10,
  no.~4, pp.~457--480, 1969.

\bibitem{261255}
H.~Sussmann and Y.~Yang, ``On the stabilizability of multiple integrators by
  means of bounded feedback controls,'' in {\em Proceedings of the 30th IEEE
  Conference on Decision and Control}, vol.~1, pp.~70--72, 1991.

\bibitem{doi:10.1016/0167-6911(93)90033-3}
Z.~Lin and A.~Saberi, ``Semi-global exponential stabilization of linear systems
  subject to “input saturation” via linear feedbacks,'' {\em Systems \&
  Control Letters}, vol.~21, no.~3, pp.~225--239, 1993.

\bibitem{doi:10.1002/rnc.4590050503}
Z.~Lin and A.~Saberi, ``A semi-global low-and-high gain design technique for
  linear systems with input saturation—stabilization and disturbance
  rejection,'' {\em International Journal of Robust and Nonlinear Control},
  vol.~5, no.~5, pp.~381--398, 1995.

\bibitem{486638}
A.~Saberi, Z.~Lin, and A.~Teel, ``Control of linear systems with saturating
  actuators,'' {\em IEEE Transactions on Automatic Control}, vol.~41, no.~3,
  pp.~368--378, 1996.

\bibitem{4610041}
B.~Zhou, G.~Duan, and Z.~Lin, ``A parametric lyapunov equation approach to the
  design of low gain feedback,'' {\em IEEE Transactions on Automatic Control},
  vol.~53, no.~6, pp.~1548--1554, 2008.

\bibitem{doi:10.1016/j.jfranklin.2020.08.025}
QianWang, Z.~Zhang, K.~Zhang, and Q.~Lin, ``Time-varying controller design for
  input saturated systems,'' {\em Journal of the Franklin Institute}, vol.~357,
  no.~15, pp.~10453--10471, 2020.

\bibitem{doi:10.1016/0167-6911(92)90001-9}
A.~R. Teel, ``Global stabilization and restricted tracking for multiple
  integrators with bounded controls,'' {\em Systems \& control letters},
  vol.~18, no.~3, pp.~165--171, 1992.

\bibitem{362853}
H.~Sussmann, E.~Sontag, and Y.~Yang, ``A general result on the stabilization of
  linear systems using bounded controls,'' {\em IEEE Transactions on Automatic
  Control}, vol.~39, no.~12, pp.~2411--2425, 1994.

\bibitem{doi:10.1016/j.automatica.2005.02.009}
Y.-S. Zhong, ``Globally stable adaptive system design for minimum phase siso
  plants with input saturation,'' {\em Automatica}, vol.~41, no.~9,
  pp.~1539--1547, 2005.

\bibitem{728887}
F.~Chaoui, F.~Giri, L.~Dugard, J.~Dion, and M.~M'saad, ``Adaptive tracking with
  saturating input and controller integral action,'' {\em IEEE Transactions on
  Automatic Control}, vol.~43, no.~11, pp.~1638--1643, 1998.

\bibitem{doi:10.1016/S0005-1098(00)00154-0}
F.~Chaoui, F.~Giri, and M.~M'Saad, ``Adaptive control of input-constrained
  type-1 plants stabilization and tracking,'' {\em Automatica}, vol.~37, no.~2,
  pp.~197--203, 2001.

\bibitem{doi:10.1016/0005-1098(94)90202-X}
G.~Feng, C.~Zhang, and M.~Palaniswami, ``Stability of input amplitude
  constrained adaptive pole placement control systems,'' {\em Automatica},
  vol.~30, no.~6, pp.~1065--1070, 1994.

\bibitem{333787}
S.~Karason and A.~Annaswamy, ``Adaptive control in the presence of input
  constraints,'' {\em IEEE Transactions on Automatic Control}, vol.~39, no.~11,
  pp.~2325--2330, 1994.

\bibitem{doi:10.1016/0005-1098(95)00059-6}
A.~M. Annaswamy and S.~Karason, ``Discrete-time adaptive control in the
  presence of input constraints,'' {\em Automatica}, vol.~31, no.~10,
  pp.~1421--1431, 1995.

\bibitem{5723705}
C.~Wen, J.~Zhou, Z.~Liu, and H.~Su, ``Robust adaptive control of uncertain
  nonlinear systems in the presence of input saturation and external
  disturbance,'' {\em IEEE Transactions on Automatic Control}, vol.~56, no.~7,
  pp.~1672--1678, 2011.

\bibitem{6975243}
M.~Chen, G.~Tao, and B.~Jiang, ``Dynamic surface control using neural networks
  for a class of uncertain nonlinear systems with input saturation,'' {\em IEEE
  Transactions on Neural Networks and Learning Systems}, vol.~26, no.~9,
  pp.~2086--2097, 2015.

\bibitem{doi:10.1080/00207721.2021.1989726}
H.~F. Qian~Wang, Yuetao~Yang and Y.~Wan, ``Input saturation: academic insights
  and future trends,'' {\em International Journal of Systems Science}, vol.~53,
  no.~6, pp.~1138--1152, 2022.

\bibitem{6875955}
J.~Ma, S.~S. Ge, Z.~Zheng, and D.~Hu, ``Adaptive nn control of a class of
  nonlinear systems with asymmetric saturation actuators,'' {\em IEEE
  Transactions on Neural Networks and Learning Systems}, vol.~26, no.~7,
  pp.~1532--1538, 2015.

\bibitem{7428909}
Y.-F. Gao, X.-M. Sun, C.~Wen, and W.~Wang, ``Observer-based adaptive nn control
  for a class of uncertain nonlinear systems with nonsymmetric input
  saturation,'' {\em IEEE Transactions on Neural Networks and Learning
  Systems}, vol.~28, no.~7, pp.~1520--1530, 2017.

\bibitem{6994268}
K.~Esfandiari, F.~Abdollahi, and H.~A. Talebi, ``Adaptive control of uncertain
  nonaffine nonlinear systems with input saturation using neural networks,''
  {\em IEEE Transactions on Neural Networks and Learning Systems}, vol.~26,
  no.~10, pp.~2311--2322, 2015.

\bibitem{doi:10.1016/j.automatica.2008.11.017}
K.~P. Tee, S.~S. Ge, and E.~H. Tay, ``Barrier lyapunov functions for the
  control of output-constrained nonlinear systems,'' {\em Automatica}, vol.~45,
  no.~4, pp.~918--927, 2009.

\bibitem{4639441}
C.~P. Bechlioulis and G.~A. Rovithakis, ``Robust adaptive control of feedback
  linearizable mimo nonlinear systems with prescribed performance,'' {\em IEEE
  Transactions on Automatic Control}, vol.~53, no.~9, pp.~2090--2099, 2008.

\bibitem{doi:10.1051/cocv:2002064}
A.~Ilchmann, E.~P. Ryan, and C.~J. Sangwin, ``Tracking with prescribed
  transient behaviour,'' {\em ESAIM: Control, Optimisation and Calculus of
  Variations}, vol.~7, p.~471–493, 2002.

\bibitem{berger2021funnel}
T.~Berger, A.~Ilchmann, and E.~P. Ryan, ``Funnel control of nonlinear
  systems,'' {\em Mathematics of Control, Signals, and Systems}, vol.~33,
  no.~1, pp.~151--194, 2021.

\bibitem{10273607}
F.~Fotiadis and G.~A. Rovithakis, ``Input-constrained prescribed performance
  control for high-order mimo uncertain nonlinear systems via reference
  modification,'' {\em IEEE Transactions on Automatic Control}, vol.~69, no.~5,
  pp.~3301--3308, 2024.

\bibitem{10004950}
P.~K. Mishra and P.~Jagtap, ``Approximation-free prescribed performance control
  with prescribed input constraints,'' {\em IEEE Control Systems Letters},
  vol.~7, pp.~1261--1266, 2023.

\bibitem{https://doi.org/10.1002/rnc.4466}
Y.~Wang, J.~Hu, J.~Li, and B.~Liu, ``Improved prescribed performance control
  for nonaffine pure-feedback systems with input saturation,'' {\em
  International Journal of Robust and Nonlinear Control}, vol.~29, no.~6,
  pp.~1769--1788, 2019.

\bibitem{10.1016/j.sysconle.2004.05.014}
A.~Ilchmann and S.~Trenn, ``Input constrained funnel control with applications
  to chemical reactor models,'' {\em Systems \& Control Letters}, vol.~53,
  no.~5, pp.~361--375, 2004.

\bibitem{10387582}
T.~Berger, ``Input-constrained funnel control of nonlinear systems,'' {\em IEEE
  Transactions on Automatic Control}, vol.~69, no.~8, pp.~5368--5382, 2024.

\bibitem{10.1016/j.arcontrol.2025.101024}
T.~Berger, A.~Ilchmann, and E.~P. Ryan, ``Funnel control — a survey,'' {\em
  Annual Reviews in Control}, vol.~60, p.~101024, 2025.

\bibitem{doi:10.2514/1.G007770}
S.~Singh, S.~R. Kumar, and D.~Mukherjee, ``Time-constrained interception with
  bounded field of view and input using barrier lyapunov approach,'' {\em
  Journal of Guidance, Control, and Dynamics}, vol.~47, no.~2, pp.~384--393,
  2024.

\bibitem{doi:10.2514/1.G009445}
S.~Kumar, S.~R. Kumar, and A.~Sinha, ``Three-dimensional trajectory tracking
  for unmanned aerial vehicles under motion constraints,'' {\em Journal of
  Guidance, Control, and Dynamics}, vol.~49, no.~1, pp.~78--96, 2026.

\end{thebibliography}
\end{document}